\documentclass[aps,prl,twocolumn,superscriptaddress]{revtex4-1}
\usepackage{graphicx,xcolor}
\usepackage{comment}

\usepackage{hyperref}

\usepackage{amsmath,amssymb}
\usepackage{braket}
\usepackage{bbm}

\usepackage{enumerate}

\usepackage{amsthm}
\theoremstyle{plain}
\newtheorem{thm}{Theorem}

\newtheorem{lemm}{Lemma}

\theoremstyle{definition}
\newtheorem{defi}{Definition}

\theoremstyle{remark}

\newcommand{\cH}{\mathcal{H}}

\newcommand{\TrH}[1]{\mathbf{T}(\mathcal{H}_{#1})}

\newcommand{\DeH}[1]{\mathbf{D}(\mathcal{H}_{#1})}
\newcommand{\BH}[1]{\mathbf{B}(\mathcal{H}_{#1})}

\newcommand{\dom}{\mathrm{dom}}

\newcommand{\tr}{\mathrm{Tr}}

\newcommand{\realn}{\mathbb{R}}
\newcommand{\cmplx}{\mathbb{C}}

\newcommand\calM{{\cal M}}

\newcommand\calP{{\cal P}}

\newcommand\calU{{\cal U}}

\newcommand{\unit}{\mathbbm{1}}

\newcommand{\Uspsp}{U}

\newcommand{\Usp}{U_{SP}}
\newcommand{\Uspo}{U_{SP}^{(1)}}
\newcommand{\Uspt}{U_{SP}^{(2)}}
\newcommand{\Uspod}{U_{SP}^{(1)\dag}}
\newcommand{\Usptd}{U_{SP}^{(2)\dag}}

\newcommand{\sfE}{\mathsf{E}}
\newcommand{\sfF}{\mathsf{F}}

\newcommand{\Spr}{{S^\prime}}
\newcommand{\Ppr}{{P^\prime}}

\newcommand{\gin}[1]{( #1 )_{g\in G}}

\newcommand{\Uone}{\mathrm{U}(1)}

\newcommand{\rmin}{\mathrm{in}}
\newcommand{\rmout}{\mathrm{out}}
\newcommand{\Ain}{A_{\mathrm{in}}}
\newcommand{\Aout}{A_{\mathrm{out}}}
\newcommand{\Bin}{B_{\mathrm{in}}}
\newcommand{\Bout}{B_{\mathrm{out}}}

\begin{document}

%\preprint{}

\title{Wigner-Araki-Yanase theorem for continuous and unbounded conserved observables}


\author{Yui Kuramochi}
\email[]{kuramochi.yui@phys.kyushu-u.ac.jp}
%\thanks{}
\affiliation{Department of Physics, Kyushu University, 744 Motooka, Nishi-ku, Fukuoka, Japan}
\author{Hiroyasu Tajima}
\affiliation{Graduate School of Informatics and Engineering,
The University of Electro-Communications,
1-5-1 Chofugaoka, Chofu, Tokyo 182-8585, Japan}


\date{\today}

\begin{abstract}
The Wigner-Araki-Yanase (WAY) theorem states that additive conservation laws imply the commutativity of exactly implementable projective measurements and the conserved observables of the system.
Known proofs of this theorem are only restricted to bounded or discrete-spectrum conserved observables of the system and are not applicable to unbounded and continuous observables like a momentum operator.
In this Letter, we present the WAY theorem for possibly unbounded and continuous conserved observables under the Yanase condition, which requires that the probe positive operator-valued measure should commute with the conserved observable of the probe system.
As a result of this WAY theorem, we show that exact implementations of the projective measurement of the position under momentum conservation and of the quadrature amplitude using linear optical instruments and photon counters are impossible.
We also consider implementations of unitary channels under conservation laws and find that the conserved observable $L_S$ of the system commute with the implemented unitary $U_S$ if $L_S$ is semi-bounded, while $U_S^\dagger L_S U_S$ can shift up to possibly non-zero constant factor if the spectrum of $L_S$ is upper and lower unbounded. 
We give simple examples of the latter case, where $L_S$ is a momentum operator.
\end{abstract}

%\keywords{}

\maketitle

\textit{\textbf{Introduction.}}---
%One of the fundamental questions of the quantum measurement theory is whether there exists limitations on the implementable measurements imposed by the physical laws.
%For example, Ozawa\rq{}s dilation theorem~\cite{:/content/aip/journal/jmp/25/1/10.1063/1.526000}
One of the fundamental findings of quantum measurement theory \cite{busch2016quantum} is the fact that the physical conservation laws restrict our ability to implement measurements.
By considering specific examples of spin measurements, Wigner~\cite{wigner1952measurement} found that additive conservation law prohibits the \textcolor{black}{projective} and repeatable measurements of an observable that does not commute with the conserved one.
He also discussed that an approximate measurement is possible if a probe state has a large coherence in the conserved quantity.
Later, Araki and Yanase~\cite{PhysRev.120.622,PhysRev.123.666} generalized Wigner's result to arbitrary repeatable projective measurements and bounded conserved observables.
The former no-go result is now called the Wigner-Araki-Yanase (WAY) theorem.

From these pioneering works by Wigner, Araki, and Yanase, many results have been published that sophisticate the WAY theorem and extend it to various directions.
The first and exciting direction is to extend the WAY theorem to a quantitative form.
Since the original WAY theorem was a qualitative theorem, many researchers, including Yanase and Ozawa, extended it to provide necessary conditions for an approximate implementation of desired measurements \cite{PhysRev.123.666,OzawaWAY,korzekwa2013,TN}.
By imposing the Yanase condition, which requires that the probe observable of the measurement model should commute with the conserved observable, it became clear that the size of the measurement device \cite{PhysRev.123.666}, the variance \cite{OzawaWAY} and quantum fluctuations \cite{korzekwa2013,TN} of the conserved quantities must be inversely proportional to the error in implementing the desired measurement.
The second direction is extending the WAY theorem to general quantum information processings beyond quantum measurements. This extension was first made as a restriction on the implementation of C-NOT gates \cite{ozawaWAY_CNOT}, extended to various limited unitary gates \cite{ozawa2003uncertainty,PhysRevA.75.032324,Karasawa_2009}, and then it was shown that for an arbitrary unitary gate \cite{TSS,TSS2}, the same restriction is given as in measurements.
This direction has been further deepened in recent years, and now extended versions of the WAY theorem are given for various objects, such as error-correcting codes \cite{TS,arxiv.2206.11086},  thermodynamic processes \cite{arxiv.2206.11086} and the toy model of black holes \cite{TS,arxiv.2206.11086}.

\begin{figure}[tb]
		\centering
		\includegraphics[width=.4\textwidth]{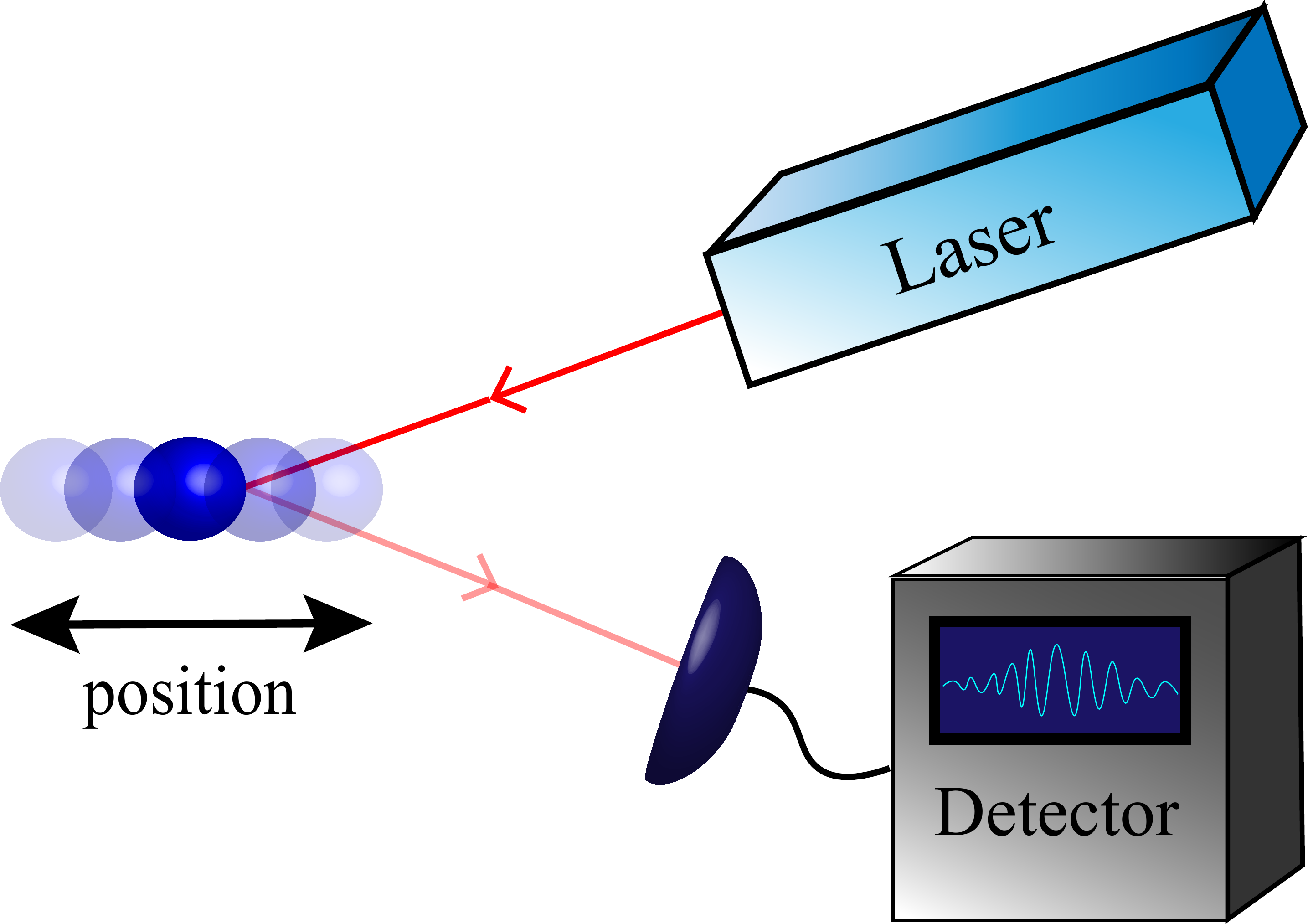}
		\caption{
		It is well known that position and momentum cannot be measured without error \textit{simultaneously}. However, in  natural settings like the above schematic, position measurements will be performed under the momentum conservation law.
		If the WAY theorem is correct for unbounded operators, the position itself cannot be measured without error under such natural settings.
		%In natural settings like the above schematic, position measurements will be performed under the momentum conservation law. If the WAY theorem is correct for unbounded operators, the position itself cannot be measured without error as long as the momentum conservation law holds.
		%It is well known that position and momentum cannot be measured without error \textit{simultaneously}. However, if the WAY theorem is correct for unbounded operators, the position itself cannot be measured without error as long as the momentum conservation law holds.
}
		\label{general_image} 
	\end{figure}

Most of the existing WAY-type results are, however, restricted to bounded conserved observables and not applicable to physically important examples in which unbounded conserved observables are common.
This problem is crucial, since if the WAY theorem is correct for unbounded operators, the position measurement without error is impossible under the momentum conservation law (see Fig.~\ref{general_image}).

Because of its importance, there has been active research on extending the WAY theorem to unbounded systems. 
However, despite previous important progress \textcolor{black}{\cite{steinshimony1971,doi:10.1063/1.525627,PhysRevLett.67.1956,ozawa1993way,PhysRevLett.106.110406,Loveridge_2020}}, this problem remains unsolved.
The extensions proved in \cite{steinshimony1971,doi:10.1063/1.525627} require some technical additional conditions, which do not hold for the position measurement under the momentum conservation.
There are also detailed accounts of the position measurements under the momentum conservation \cite{PhysRevLett.67.1956,PhysRevLett.106.110406}.
Particularly in Ref. \cite{PhysRevLett.106.110406}, a trade-off relation is obtained for the accuracy of the position measurement and the momentum coherence of the probe under the momentum conservation.
However, in the derivation of the trade-off, as pointed out in \cite{Loveridge_2020}, issues related to domains of unbounded operators are ignored.
Moreover, even if the trade-off relation was valid, it would not imply the impossibility of the exact implementation of the position measurement under momentum conservation because we can prepare a probe state with a divergent momentum coherence and the trade-off relation gives the trivial inequality $0\geq 0$ in this case.

Here we give a positive answer to this question: we present the WAY theorem for general unbounded and continuous conserved observables under the Yanase condition, which is a basic condition introduced in Refs. \cite{PhysRev.123.666,OzawaWAY} and used in \cite{PhysRevLett.106.110406}.
We also consider the unitary channel implementation, and show that a similar theorem holds in that case. For the unitary channel implementation without error under the conservation law, we show that the implemented unitary $U_S$ and the conserved observable $L_S$ must commute, except for a very limited scenario that $L_S$ is upper and lower unbounded and the change of $L_S$ by $U_S$ is a constant shift: $U^\dagger_S L_S U_S=L_S+\gamma \unit $~\footnote{}.

\textit{\textbf{Notation and definitions.}}---
In this letter, the Hilbert space of a quantum system $S$ is denoted by $\cH_S$, which may be finite or infinite-dimensional.
The unit \textcolor{black}{operator} and sets of bounded and trace-class \textcolor{black}{operators} on a Hilbert space $\cH_S$ are respectively denoted by $\unit_S, \, \BH{S}$, and $\TrH{S}$.
A non-negative operator $\rho_S \in \TrH{S}$ with a unit trace is called a density operator, which corresponds to a quantum state.
The set of density operators on $\cH_S$ is denoted by $\DeH{S}$.

For a linear map $ \Lambda \colon \TrH{A} \to \TrH{B}$ that is bounded with respect to the trace-norms on $ \TrH{A}$ and $\TrH{B}$, the adjoint map $\Lambda^\dag \colon \BH{B} \to \BH{A}$ is well-defined by $\tr[\Lambda(\rho_A) b]  = \tr[\rho_A \Lambda^\dag (b)] \, (\rho_A \in \TrH{A}; \, b\in \BH{B})$, where $\tr [\cdot]$ denotes the trace.
A linear map $\Lambda \colon \TrH{A} \to \TrH{B}$ is called a \textit{quantum channel} if $\Lambda$ is trace-preserving and \textcolor{black}{$\Lambda^\dag$} is completely positive (CP)~\cite{1955stinespring,paulsen_2003}.
The \textcolor{black}{map} $\Lambda^\dag$ represents the channel in the Heisenberg picture.

A triple $(\Omega, \Sigma, \sfE_S)$ is called a \textit{positive operator-valued measure} (POVM) \cite{davies1976quantum,holevo2001statistical} on $\cH_S$ if $\Sigma$ is a $\sigma$-algebra on the set $\Omega$ and $\sfE_S \colon \Sigma \to \BH{S}$ satisfies (i) $\sfE_S(X) \geq 0 \, (X\in \Sigma)$, 
(ii) $\sfE_S(\varnothing)=0, \, \sfE_S (\Omega) = \unit_S$, 
and (iii) $  \sfE_S(\cup_k X_k) = \sum_k \sfE(X_k)$ in the weak operator topology~\cite{reed1972methods,prugovecki1982quantum} for any disjoint sequence \textcolor{black}{$(X_k) \subseteq \Sigma$.}
A POVM $(\Omega, \Sigma, \sfE_S)$ is called a projection-valued measure (PVM) if each $\sfE_S(X) \, (X\in \Sigma)$ is a projection.
A POVM $(\Omega, \Sigma, \sfE_S)$ on $\cH_S$ describes the outcome statistics of a general measurement process so that the outcome probability measure when the state is prepared in $\rho_S \in \DeH{S}$ is given by $\Sigma \ni X \mapsto \tr [\rho_S \sfE_S(X)]$.

We now consider the implementation of a quantum channel $ \Psi \colon \TrH{S} \to \TrH{\Spr}$ by a system-environment model.
Here a tuple $ (\cH_P, \cH_{\Spr} , \cH_{P^\prime} , \rho_P , \Uspsp)$ is called a \textit{system-environment model} if $P, \, \Spr$, and $\Ppr$ are quantum systems, $\rho_P \in \DeH{P}$ is a density operator on $P$ called the probe state, and $\Uspsp \colon \cH_S \otimes \cH_P \to \cH_\Spr \otimes \cH_\Ppr$ is a unitary operator.
The system-environment model $ (\cH_P, \cH_\Spr, \cH_{P^\prime} , \rho_P , \Uspsp)$ is said to \textit{implement} $\Psi $ if
\begin{equation}
	\Psi (\rho_S) = \tr_\Ppr [\Uspsp (\rho_S \otimes \rho_P) \Uspsp^\dag]
	\quad (\rho_S \in \DeH{S}),
	\label{eq:implementc}
\end{equation} 
%or equivalently in the Heisenberg picture $\Psi^\dag (b) = \tr_P [(\unit_S \otimes \rho_P) \Uspsp^\dag (b\otimes \unit_\Ppr) \Uspsp ]$ $(b\in \BH{\Spr})$, 
where $\tr_A[\cdot]$ denotes the partial trace over a system $A$ and the dagger denotes the \textcolor{black}{adjoint.}
The condition \eqref{eq:implementc} says that the channel $\Psi$ is realized if we first prepare the system $S$ in an arbitrary state $\rho_S$ and $P$ in the fixed probe state $\rho_P$, then they interact according to the unitary $\Uspsp$, and finally discard the output probe system $\Ppr$ (see Fig.~\ref{fig:1}).

\begin{figure} 
\centering
\includegraphics[width=5.5cm,clip]{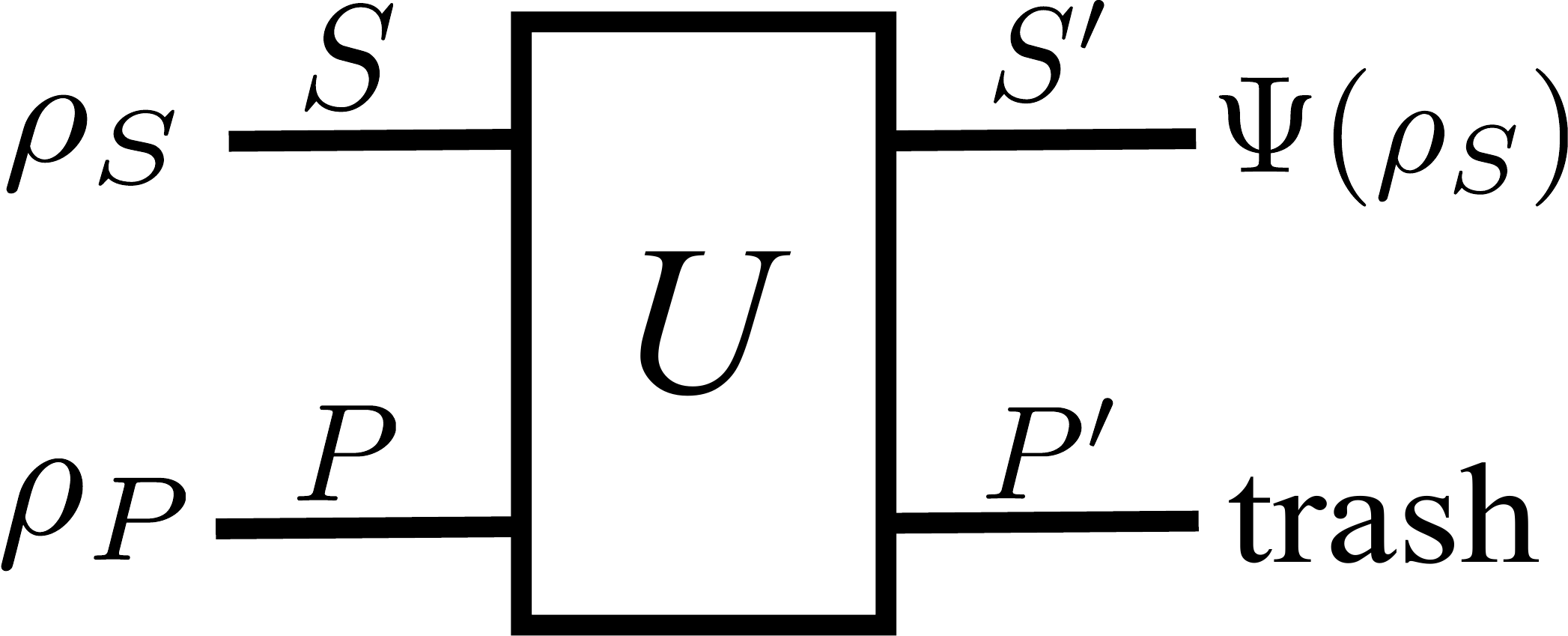}
% {\large
% \Qcircuit @C=1em @R=1.2em { \lstick{\rho_S} &    \ustick{S}\qw& \multigate{1}{U^\dag} & \ustick{\Spr} \qw  &   \rstick{\Psi(\rho_S)} \qw \\ 
% \lstick{\rho_P} &    \dstick{P} \qw & \ghost{U^\dag}& \dstick{\Ppr} \qw  & \rstick{\text{trash}} \qw}
%  }
\caption{A system-environment model that implements a channel $\Psi$.}
\label{fig:1}
\end{figure}

\begin{figure} 
\centering
% {\large
% \Qcircuit @C=1em @R=1.2em { \lstick{\rho_S} &    \ustick{S}\qw& \multigate{1}{U^\dag} & \ustick{\Spr}  \qw &   \rstick{\text{trash}} \qw \\ 
% \lstick{\rho_P} &    \dstick{P} \qw & \ghost{U^\dag}& \dstick{\Ppr} \qw  & \measureD{\sfF_{\Ppr}} \qw}
% }
\includegraphics[width=5.5cm,clip]{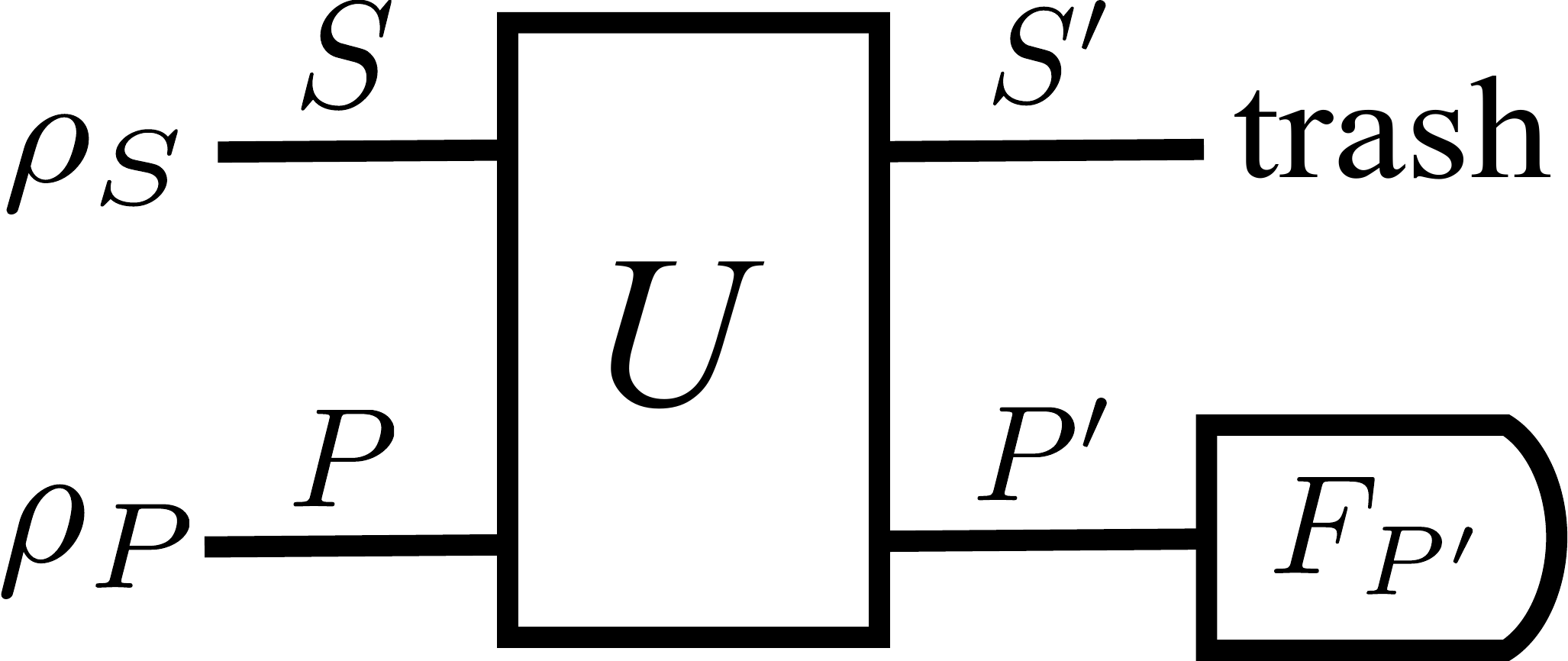}
% $=$
% \includegraphics[width=3cm,clip]{fig2_2.pdf}
\caption{A measurement model that implements a POVM $\sfE_S$.}
\label{fig:2}
\end{figure}

The implementation of a measurement by a measurement model is defined in a similar way as follows.
A tuple $\mathbb{M}= ( \cH_P, \cH_{S^\prime} , \cH_{P^\prime} , \rho_P , \Uspsp , (\Omega, \Sigma ,\sfF_\Ppr))$ is called a \textit{measurement model} if $ ( \cH_P, \cH_{S^\prime} , \cH_{P^\prime} , \rho_P , \Uspsp)$ is a system-environment model and $(\Omega, \Sigma ,\sfF_\Ppr)$ is a POVM on $\cH_\Ppr$.
The measurement model $\mathbb{M}$ is said to \textit{implement} a POVM $(\Omega , \Sigma , \sfE_S)$ on $\cH_S$ if 
%\begin{equation}
%	\tr[\rho_S \sfE_S(X)] = \tr[\Uspsp (\rho_S \otimes \rho_P) \Uspsp^\dag (\unit_\Spr \otimes \sfF_\Ppr (X))]
%\end{equation}
%$(\rho_S\in \DeH{S} ;\, X\in \Sigma)$, 
\begin{equation}
	\sfE_S(X) = \tr_P [(\unit_S \otimes \rho_P) \Uspsp^\dag (\unit_\Spr \otimes \sfF_\Ppr (X)) \Uspsp]
\end{equation}
for all $X\in \Sigma$ (see Fig.~\ref{fig:2}).

Let $L_A$ be a possibly unbounded self-adjoint operator \cite{reed1972methods,prugovecki1982quantum} on a Hilbert space $\cH_A$
and let $\dom (L_A) \subseteq \cH_A$ denote the domain of $L_A$.
A bounded operator $a \in \BH{A}$ is said to \textit{commute} with $L_A$ if $a$ commutes with the spectral measure \cite{reed1972methods,prugovecki1982quantum} of $L_A$.
If $U_A \in \BH{A}$ is \textcolor{black}{unitary,} then $U_A$ commutes with $L_A$ if and only if $L_A = U_A^\dag L_A U_A$, where the domain of the self-adjoint operator $U_A^\dag L_A U_A$ \textcolor{black}{is} $U_A^\dag \dom(L_A)$.

The \textit{spectrum} $\sigma (L_A)$ of a self-adjoint operator $L_A$ on $\cH_A$ is the set of $\lambda \in \cmplx$ such that the operator $L_A - \lambda \unit_A$ has no bounded \textcolor{black}{inverse.}
\textcolor{black}{The spectrum $\sigma(L_A)$} is a closed subset of the reals $\realn$ and, if $\cH_A$ is finite-dimensional, \textcolor{black}{coincides} with the set of the eigenvalues of $L_A$.
A self-adjoint operator $L_A$ is said to be \textit{semi-bounded} (respectively, \textit{unbounded}) if $\sigma (L_A)$ is an upper or lower bounded (respectively, unbounded) subset of $\realn$.
For example, the quantum harmonic oscillator Hamiltonian is \textcolor{black}{unbounded but still semi-bounded.}

\textit{\textbf{Main results.}}---
Now we state the main results of this letter:
\begin{thm}[WAY theorem for projective measurements] \label{thm:main1}
Let 
$
( \cH_P, \cH_{S^\prime} , \cH_{P^\prime} , \rho_P,  \Uspsp , (\Omega, \Sigma, \sfF_\Ppr))
$
be a measurement model that implements a PVM $(\Omega ,\Sigma, \sfE_S)$ on $\cH_S$. Suppose that there are (possibly unbounded) self-adjoint operators $L_S,\, L_P, \, L_{S^\prime}$, and $L_{P^\prime}$ that act respectively on $\cH_S$, $\cH_P$, $\cH_{S^\prime}$, and $\cH_{P^\prime}$ and satisfy the conservation law
\begin{equation}
	\Uspsp^\dag L_{\Spr \Ppr}  \Uspsp = L_{SP},
	\label{eq:cons}
\end{equation}
where $L_{SP} := L_S \otimes \unit_P + \unit_S \otimes L_P $ and $L_{\Spr \Ppr} := L_{S^\prime} \otimes \unit_{P^\prime} + \unit_{S^\prime} \otimes L_\Ppr$.
We also assume the Yanase condition that $\sfF_\Ppr (X)$ commutes with $L_\Ppr$ for every $X \in \Sigma$.
Then $\sfE_S(X)$ commutes with $L_S$ for every $X\in \Sigma$.
\end{thm}

\begin{thm}[WAY theorem for unitary channels] \label{thm:main2}
Let $U_{S\to \Spr} \colon \cH_S \to \cH_\Spr$ be a unitary operator between Hilbert spaces $\cH_S$ and $\cH_\Spr$, let $\mathcal{U}_{S\to \Spr} \colon \TrH{S} \to \TrH{\Spr}$ be the unitary channel defined by $\mathcal{U}_{S\to \Spr} (\rho_S) := U_{S\to \Spr} \rho_S U_{S\to \Spr}^\dag$ $(\rho_S \in \TrH{S})$, and let $(\cH_S, \cH_P, \cH_{S^\prime} , \cH_{P^\prime} , \rho_P , \Uspsp)$ be a system-environment model that implements $\mathcal{U}_{S\to \Spr}$.
Suppose that there are (possibly unbounded) self-adjoint operators $L_S, \, L_P, \, L_{S^\prime}$, and $L_{P^\prime}$ that act respectively on $\cH_S, \, \cH_P, \, \cH_{S^\prime}$, and $\cH_{P^\prime}$ and satisfy the conservation law~\eqref{eq:cons}.
Then there exists a real number $\gamma \in \realn$ such that 
\begin{equation}
	U_{S\to \Spr}^\dag L_\Spr U_{S\to \Spr} = L_S + \gamma \unit_S .
	\label{eq:main2-1}
\end{equation}
Moreover, if $\cH_S = \cH_\Spr$ and $L_S = L_\Spr$ hold and $L_S$ is semi-bounded, then $U_S := U_{S\to \Spr}$ commutes with $L_S$.
\end{thm}
The latter part of Theorem~\ref{thm:main2} can be immediately proved from the former part as follows.
Assume $\cH_S = \cH_\Spr$ and $L_S = L_\Spr$.
Then Eq.~\eqref{eq:main2-1} implies
\begin{equation}
	U_{S}^\dag L_S U_{S} = L_S + \gamma \unit_S . 
	\label{eq:spec}
\end{equation}
Suppose that $\sigma (L_S)$ is lower bounded and let $\lambda_{\mathrm{min}}(L_S) \in \realn$ denote the finite infimum of the spectrum $\sigma (L_S)$.
Since the spectra of $L_S$ and $U_{S}^\dag L_S U_{S}$ coincide, Eq.~\eqref{eq:spec} implies $\lambda_{\mathrm{min}}(L_S)  = \lambda_{\mathrm{min}}(L_S) + \gamma$
and therefore $\gamma = 0$.
Hence $U_S$ commutes with $L_S$.
The claim is similarly proved by considering the supremum of $\sigma (L_S)$ when $L_S$ is upper bounded.

Theorems~\ref{thm:main1} and \ref{thm:main2} can be proved by using the notion of the multiplicative domains of unital CP maps~\cite{choi1974,paulsen_2003}.
This notion is recently used in \cite{mohammady2022} \textcolor{black}{to derive} WAY-type trade-off relations for bounded observables.
\textcolor{black}{In the proof, arguments on the topological group $\realn$ and its unitary representations are also essential that derive statements valid for all $t\in \realn$ from those valid only for restricted $t$.}
%, in the case of Theorem~\ref{thm:main1}, the local commutativity of $e^{itL_S}$ with $\sfE_S$ to all $t\in \realn$.}
We also remark that Theorems~\ref{thm:main1} and \ref{thm:main2} can be generalized to general continuous symmetries described by connected topological groups~\cite{pontryagin1986topological,higgins_1975}.
All the details of the proofs, including the generalization to \textcolor{black}{continuous symmetries}, are given in the Supplemental Material.

\textit{\textbf{Applications of the WAY theorem for projective measurements.}}---
Now we see two applications of Theorem~\ref{thm:main1} \textcolor{black}{which show that some kinds of measurements are not implementable.}

The first one is the position measurement under the momentum conservation~\cite{PhysRevLett.67.1956,PhysRevLett.106.110406}.
Since the position and momentum operators of a $1$-dimensional quantum particle are noncommutative in the sense that their spectral measures do not commute, it immediately follows from Theorem~\ref{thm:main1} that no measurement model satisfying the momentum conservation and the Yanase condition can implement the \textcolor{black}{projective} position measurement of the particle.
This gives a positive answer to the open question in \cite{PhysRevLett.106.110406}.

The next one is the \textcolor{black}{projective} measurement of a quadrature amplitude of a single-mode optical field by using beam splitters, phase shifters, and photon counters.
We consider fixed-frequency optical fields and denote by $\hat{a}_A$ the annihilation operator acting on the Hilbert space $\cH_A$ of a mode $A$.
\textcolor{black}{In this situation, the accurate implementation of the projective measurement of the quadrature amplitude operator $\hat{q}_S = (\hat{a}_S + \hat{a}_S^\dag)/2$ is important in continuous-variable (CV) quantum technologies like CV quantum key distribution~\cite{Hirano_2017} or CV quantum teleportation~\cite{furusawa1998}.}
\textcolor{black}{However, since the quadrature amplitude operator $\hat{q}_S = (\hat{a}_S + \hat{a}_S^\dag)/2$ does not commute with the number operator $L_S = \hat{n}_S = \hat{a}_S^\dag \hat{a}_S$, Theorem~\ref{thm:main1} implies that the errorless projective measurement $\sfE_S$ of $\hat{q}_S$ is not implementable by any measurement model satisfying the conservation law of the total photon number and the Yanase condition.}

\textcolor{black}{To see the detail of the above, let us introduce a measurement model of the passive optical operations (see Fig.~\ref{fig:4}) and how Theorem~\ref{thm:main1} works on this model.}
A two-mode passive optical unitary ${\color{black}V} \colon \cH_{A_\rmin} \otimes \cH_{B_\rmin} \to \cH_{A_\rmout} \otimes \cH_{B_\rmout}$ is a unitary such that ${\color{black}V^\dag \hat{a}_{A_\rmout, B_\rmout}V}$ is a linear combination of $\hat{a}_{A_\rmin}$ and $\hat{a}_{B_\rmin}$ and energy (photon number) conservation
\begin{equation}
	{\color{black}V^\dag (\hat{n}_{\Aout}  + \hat{n}_{\Bout}) V}
	= \hat{n}_{\Ain} +\hat{n}_{\Bin}
	\label{eq:pn}
\end{equation}
holds.
Here we abbreviated the identities and tensors.
The Hilbert spaces $\cH_P = \cH_{P_1}\otimes \dots \otimes \cH_{P_N}, \, \cH_{\Spr} = \cH_{\Spr_1} \otimes \dots \otimes \cH_{\Spr_M},  \, \cH_{\Ppr} = \cH_{\Ppr_1} \otimes \dots \otimes \cH_{\Ppr_{N-M+1}}$ are finite tensor products of single mode Hilbert spaces and
the \textcolor{black}{total} unitary $\Uspsp \colon \cH_S \otimes \cH_P \to \cH_{\Spr} \otimes \cH_\Ppr$ is a finite compositions of passive optical unitaries satisfying Eq.~\eqref{eq:pn}.
\textcolor{black}{We assume that the} probe POVM $(\Omega , \Sigma , \sfF_\Ppr (\cdot))$ on $\cH_\Ppr$ commutes with the outcome photon number operators $\hat{n}_{\Ppr_1},\dots , \hat{n}_{\Ppr_{N-M+1}}$ \textcolor{black}{so that the Yanase condition holds.
For} example, if the probe measurement $\sfF_\Ppr$ is realized by post-processing the outcomes of photon-counting measurements on the modes $\Ppr_{1} , \dots , \Ppr_{N-M+1}$, this condition \textcolor{black}{holds.}

\textcolor{black}{Due to Theorem~\ref{thm:main1}, the above model cannot implement the projective measurement $\sfE_S$ of $\hat{q}_S$. 
To see that, let us put the conserved observables in Theorem~\ref{thm:main1} as $L_S := \hat{n}_S, \, L_P := \sum_k \hat{n}_{P_k} ,\, L_\Spr := \sum_k \hat{n}_{\Spr_k}, \, L_\Ppr := \sum_k \hat{n}_{\Ppr_k}$. Then the conservation law \eqref{eq:cons} holds, which is in this case the total photon number conservation $U^\dag (\hat{n}_S + \hat{n}_P)U = \hat{n}_\Spr + \hat{n}_\Ppr $. 
Moreover, since $\hat{q}_S$ does not commute with $\hat{n}_S=L_S$, the projective measurement $\sfE_S$ of $\hat{q}_S$ also does not commute with $L_S$. 
%Moreover, the projective measurement $\sfE_S$ of $\hat{q}_S$ does not commute with $L_S = \hat{n}_S .$
Therefore Theorem~\ref{thm:main1} prohibits the implementation of $\sfE_S$.}
We remark that we do not require any condition on the probe \textcolor{black}{state} $\rho_P$.

\begin{figure} 
\centering
\includegraphics[width=8.8cm,clip]{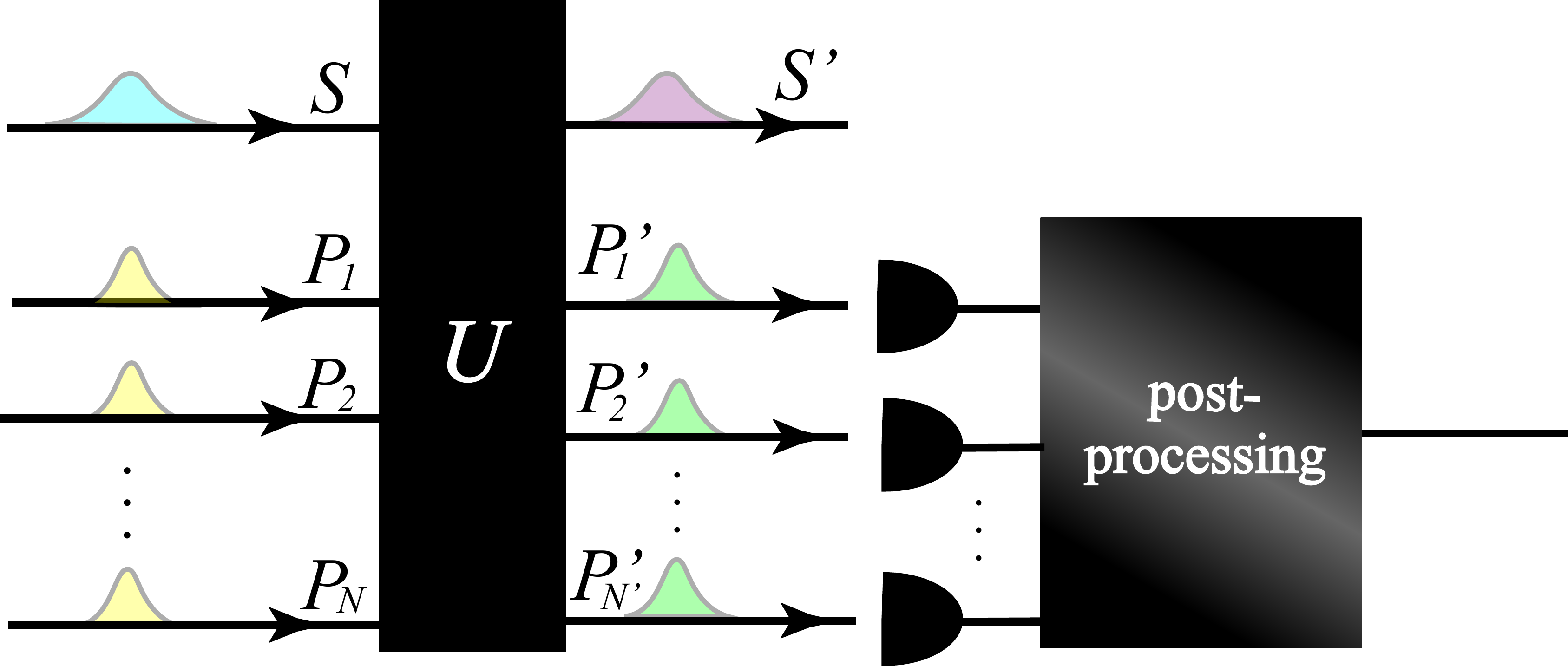}
% \begin{tikzpicture}[>=stealth]
% \tikzset{pics/.cd,
% collector/.style={code={
% \draw[fill=gray!20] (0,0.5) arc(90:-90:0.75cm and 0.5cm) -- cycle;}},
% splitter/.style={code={\draw[ultra thick] (#1:{sqrt(1/2)}) --
% (#1+180:{sqrt(1/2)});}},splitter/.default=135
% };
% \draw[->] (-1,0) node[left]{$\alpha_i$} -- (4,0) pic[right=1ex]{collector};
% \draw[->] (0,1) node[above]{$\beta_1$} -- (0,-2.5) pic[below=1ex,rotate=-90]{collector};
% \draw[->] (3,1) node[above]{$\beta_2$} -- (3,-1) pic[below=1ex,rotate=-90]{collector};
% \draw[->] (-1,-1.5) node[left]{$\beta_3$} -- (1,-1.5) pic[right=1ex]{collector};
% %\path (0,0) pic{splitter} node[above right]{$\eta_1$}
% % (3,0) pic{splitter} node[above right]{$\eta_2$}
% %  (0,-1.5) pic{splitter} node[above right]{$\eta_3$};
% \end{tikzpicture}
\caption{Measurement model with passive optical operations and photon-counting measurements.}
\label{fig:4}
\end{figure}

We can still realize \textit{approximate} measurement of $\hat{q}_S$ by the balanced homodyne detection~\cite{yuen1978quantum,Tyc_2004}, in which the signal optical field is mixed with a strong local oscillator (LO) \textcolor{black}{field} by a half beam splitter and the properly normalized difference of the photocounts of the output fields is recorded.
The measurement model of the homodyne detection apparently satisfies the above assumptions and hence does not implement the \textcolor{black}{projective} measurement of $\hat{q}_S$.

On the other hand, it can be shown~\cite{yuen1978quantum} that if we prepare the probe LO state as a coherent state $\ket{\beta_{\mathrm{LO}}}_P = e^{-|\beta_{\mathrm{LO}}|^2/2}\sum_{n=0}^\infty \frac{\beta_{\mathrm{LO}}^n}{\sqrt{n!}} \ket{n}_P \, (\beta_{\mathrm{LO}}   \in \realn)$, where $\ket{n}_P$ denotes the photon number eigenstate of the probe LO field, then, for every initial state $\rho_S$, the probability distribution of the homodyne measurement converges in distribution to that of the \textcolor{black}{projective} measurement of $\hat{q}_S$ in the strong LO limit $\beta_{\mathrm{LO}} \to \infty$.
This is in accordance with the \lq\lq{}positive part\rq\rq{} of the original WAY arguments, since strong LO means a large spread of $\ket{\beta_{\mathrm{LO}}}_P$ in the photon number basis.
We should still be careful about the state-wise nature of the convergence that results from the unboundedness of the conserved observable $\hat{n}_S$.
For example, if we prepare the input state as a coherent state $\ket{\alpha_S}_S$ and $|\alpha_S|$ is comparable with the LO amplitude $\beta_{\mathrm{LO}}$, the outcome distribution is far from that of the \textcolor{black}{projective} measurement of $\hat{q}_S$.

\textit{\textbf{Examples of implementations of unitary channels.}}---
We now give two examples of implementations of a unitary channel in which the constant term $\gamma \unit_S$ in Eq.~\eqref{eq:main2-1} is non-zero.

In the models, the final systems $\Spr$ and $\Ppr$ are, respectively, the same as the initial systems $S$ and $P$.
We take $1$-dimensional quantum particles as the system and probe systems so that the Hilbert spaces $\cH_S = \cH_\Spr$ and $\cH_P = \cH_\Ppr$ are both \textcolor{black}{the} space $L^2(\realn )$ of square-integrable functions on $\realn$.
Let $\hat{x}_{\alpha}$ and $\hat{p}_{\alpha} \, (\alpha = S, P)$ denote respectively the position and momentum operators of the system $\alpha$, which satisfy the Weyl relation
\begin{equation}
	e^{it \hat{x}_\alpha} e^{is \hat{p}_\alpha} = e^{-ist}  e^{is \hat{p}_\alpha}  e^{it \hat{x}_\alpha}  \quad (s,t \in \realn; \, \alpha = S, P),
	\label{eq:weyl}
\end{equation}
where $\hbar$ is set to $1$.
We fix an arbitrary real number $\gamma \neq 0$ and give two implementation models of the unitary channel $\mathcal{U}_S (\rho_S) = U_S \rho_S U_S^\dag$ with $U_S := e^{i\gamma \hat{x}_S}$.
We put $L_S = L_\Spr = \hat{p}_S$ and $L_P = L_\Ppr = \hat{p}_P$.
Then from \textcolor{black}{Eq.~\eqref{eq:weyl}} we can \textcolor{black}{see} that $U_S^\dag L_S  U_S = L_S + \gamma \unit_S$ holds.

In the first example, we take the following total unitary
\begin{equation}
	\Uspo := e^{i\gamma \hat{x}_S} \otimes e^{-i\gamma \hat{x}_P}.
	\label{eq:usp1}
\end{equation}
Then from \eqref{eq:weyl} this unitary satisfies the momentum conservation law 
\begin{equation}
	\Usp^{(1)\dag} (\hat{p}_S + \hat{p}_P) \Uspo = \hat{p}_S + \hat{p}_P,
	\label{eq:consp1}
\end{equation}
where we omitted the tensors and units.
Moreover, for an arbitrary probe state $\rho_P \in \DeH{P}$ we have
\begin{equation}
	\tr_P [ \Uspo (\rho_S \otimes \rho_P)  \Uspod ]
	= U_S \rho_S U_S^\dag 
	\label{eq:model1}
\end{equation}
$(\rho_S \in \DeH{S})$.
This shows that the system-environment model $(\cH_S , \cH_P, \cH_S, \cH_P , \rho_P , \Uspo)$ satisfies all the assumptions of Theorem~\ref{thm:main2} together with \eqref{eq:main2-1} with non-zero $\gamma$.

There is another example of an implementation of the unitary channel $\mathcal{U}_S$ in which the total unitary is not in product form.
For simplicity, we assume $\gamma >0$ and define the total unitary
\begin{equation}
	\Uspt := e^{i\gamma \hat{x}_S} \otimes e^{-i\gamma \hat{x}_P} 1_X (\hat{p}_P)  + \unit_S \otimes 1_{\realn \setminus X} (\hat{p}_P),
	\label{eq:uspt}
\end{equation}
which is not in product form.
Here,
\begin{equation}
	1_A (\lambda) :=
	\begin{cases}
		1 & (\lambda \in A); \\
		0 & (\lambda \not\in A)
	\end{cases}
\end{equation}
is the indicator function of a subset $A\subseteq \realn$ %so that $A \mapsto 1_A (\hat{p}_P)$ is the spectral measure of $\hat{p}_P$ 
and $X\subseteq \realn$ is a measurable set such that $X+\gamma := \{ x+\gamma : x\in X\} = X$ and neither $X$ nor $\realn \setminus X$ is a null set.
For definiteness, we take as 
\textcolor{black}{$
	X = \bigcup_{n\colon \text{integer}} [\gamma n - \gamma /3 , \gamma n + \gamma /3] .
$}
Then since
$
	e^{i\gamma \hat{x}_P} 1_X (\hat{p}_P) e^{-i\gamma \hat{x}_P}
	= 1_X  (\hat{p}_P - \gamma ) = 1_{X + \gamma} (\hat{p}_P) = 1_X(\hat{p}_P)
$,
that is, $e^{i\gamma \hat{x}_P}$ and $1_X (\hat{p}_P)$ commute, the operator $\Uspt$ \textcolor{black}{in} Eq.~\eqref{eq:uspt} is unitary.
Moreover from  
\begin{widetext}
\begin{align}
	\Usptd (e^{it\hat{p}_S} \otimes e^{it\hat{p}_P}) \Uspt
	&= e^{-i\gamma \hat{x}_S} e^{it\hat{p}_S} e^{i\gamma \hat{x}_S} \otimes e^{i\gamma \hat{x}_P} e^{it\hat{p}_P} e^{-i\gamma \hat{x}_P} 1_X(\hat{p}_P)
	+ e^{it\hat{p}_S} \otimes e^{it\hat{p}_P} 1_{\realn \setminus X}(\hat{p}_P)
	\\
	&=e^{it\hat{p}_S} \otimes e^{it\hat{p}_P} \quad (t\in \realn),
\end{align}
\end{widetext}
the momentum conservation 
\begin{equation}
	\Usptd (\hat{p}_S + \hat{p}_P) \Uspt = \hat{p}_S + \hat{p}_P
	\label{eq:consp2}
\end{equation}
holds.
If we take a state $\rho_P \in \DeH{P}$ supported by the projection $1_X(\hat{p}_P)$, 
%(i.e.\ $ 1_X(\hat{p}_P) \rho_P 1_X(\hat{p}_P) = \rho_P$), 
then we have $ \Uspt (\rho_S \otimes \rho_P)  \Usptd =  \Uspo (\rho_S \otimes \rho_P)  \Uspod  \, (\rho_S \in \DeH{S})$ and therefore from Eq.~\eqref{eq:model1} we can  \textcolor{black}{see} that the system-environment model $( \cH_P, \cH_S, \cH_P , \rho_P , \Uspt)$ implements \textcolor{black}{$\mathcal{U}_S.$}

\textit{\textbf{Conclusion.}}---
\textcolor{black}{We investigated measurement implementations under conservation laws of unbounded observables and} established the WAY theorem for projective measurements under the Yanase condition.
Applications of this WAY theorem revealed that the \textcolor{black}{projective} measurements of the position and the quadrature amplitude are incompatible with the conservation of the momentum and the photon number, respectively.
\textcolor{black}{It is still open whether the original WAY theorem~\cite{PhysRev.120.622} (or Theorem~8.1 of \cite{:/content/aip/journal/jmp/25/1/10.1063/1.526000}) for \textit{repeatable} measurement models can be generalized to unbounded conserved observables.}

We also considered implementation of unitary channels under conservation laws and found that the implemented unitary \textcolor{black}{commutes} with the conserved observable if it is semi-bounded, while the conserved observable can shift up to a constant factor if the conserved \textcolor{black}{observable} is upper and lower unbounded.
The former case in finite dimensions can be immediately derived from the more general trade-off  \textcolor{black}{relation~\cite{TSS2}, while}
the latter case is  \textcolor{black}{essentially infinite-dimensional} and cannot be expected from the finite-dimensional existing works.

\textcolor{black}{Our work has several possible directions of future extensions. 
%One such possibility is the generalization to the state-dependent scenario, e.g.\ energy-constrained states.
One such possibility is the generalization to the state-dependent scenario (e.g.\ energy-constrained states), while our results are restricted to state-independent case.
%A physically interesting example of such problem is to consider energy-constrained states.
Another possible extension is to consider approximate implementations.}
This work concerns only the extreme case of \textit{exact} implementations of projective measurements or unitary channels.
On the other hand, as mentioned in the introduction, results on approximate implementations of measurements and unitary gates have been actively studied in recent years. With few exceptions, these have not been extended to infinite-dimensional systems. (See the brief review in Supplementary Materials~\footnote{}.)
It is an interesting future direction to extend these results to unbounded observables.

\begin{acknowledgments}
YK acknowledges the support by JSPS Grant-in-Aid for Early-Career Scientists No.~JP22K13977.
HT acknowledges the supports by JSPS Grants-in-Aid for Scientific Research No.~JP19K14610 and No.~JP22H05250, JST PRESTO (Grant No.~JPMJPR2014), and  JST MOONSHOT (Grant No.~JPMJMS2061).
\end{acknowledgments}

%\bibliography{abbr}

\begin{thebibliography}{37}%
\makeatletter
\providecommand \@ifxundefined [1]{%
 \@ifx{#1\undefined}
}%
\providecommand \@ifnum [1]{%
 \ifnum #1\expandafter \@firstoftwo
 \else \expandafter \@secondoftwo
 \fi
}%
\providecommand \@ifx [1]{%
 \ifx #1\expandafter \@firstoftwo
 \else \expandafter \@secondoftwo
 \fi
}%
\providecommand \natexlab [1]{#1}%
\providecommand \enquote  [1]{``#1''}%
\providecommand \bibnamefont  [1]{#1}%
\providecommand \bibfnamefont [1]{#1}%
\providecommand \citenamefont [1]{#1}%
\providecommand \href@noop [0]{\@secondoftwo}%
\providecommand \href [0]{\begingroup \@sanitize@url \@href}%
\providecommand \@href[1]{\@@startlink{#1}\@@href}%
\providecommand \@@href[1]{\endgroup#1\@@endlink}%
\providecommand \@sanitize@url [0]{\catcode `\\12\catcode `\$12\catcode
  `\&12\catcode `\#12\catcode `\^12\catcode `\_12\catcode `\%12\relax}%
\providecommand \@@startlink[1]{}%
\providecommand \@@endlink[0]{}%
\providecommand \url  [0]{\begingroup\@sanitize@url \@url }%
\providecommand \@url [1]{\endgroup\@href {#1}{\urlprefix }}%
\providecommand \urlprefix  [0]{URL }%
\providecommand \Eprint [0]{\href }%
\providecommand \doibase [0]{http://dx.doi.org/}%
\providecommand \selectlanguage [0]{\@gobble}%
\providecommand \bibinfo  [0]{\@secondoftwo}%
\providecommand \bibfield  [0]{\@secondoftwo}%
\providecommand \translation [1]{[#1]}%
\providecommand \BibitemOpen [0]{}%
\providecommand \bibitemStop [0]{}%
\providecommand \bibitemNoStop [0]{.\EOS\space}%
\providecommand \EOS [0]{\spacefactor3000\relax}%
\providecommand \BibitemShut  [1]{\csname bibitem#1\endcsname}%
\let\auto@bib@innerbib\@empty
%</preamble>
\bibitem [{\citenamefont {Busch}\ \emph {et~al.}(2016)\citenamefont {Busch},
  \citenamefont {Lahti}, \citenamefont {Pellonp{\"a}{\"a}},\ and\ \citenamefont
  {Ylinen}}]{busch2016quantum}%
  \BibitemOpen
  \bibfield  {author} {\bibinfo {author} {\bibfnamefont {P.}~\bibnamefont
  {Busch}}, \bibinfo {author} {\bibfnamefont {P.~J.}\ \bibnamefont {Lahti}},
  \bibinfo {author} {\bibfnamefont {J.-P.}\ \bibnamefont {Pellonp{\"a}{\"a}}},
  \ and\ \bibinfo {author} {\bibfnamefont {K.}~\bibnamefont {Ylinen}},\
  }\href@noop {} {\emph {\bibinfo {title} {{Quantum Measurement}}}}\ (\bibinfo
  {publisher} {Springer},\ \bibinfo {year} {2016})\BibitemShut {NoStop}%
\bibitem [{\citenamefont {Wigner}(1952)}]{wigner1952measurement}%
  \BibitemOpen
  \bibfield  {author} {\bibinfo {author} {\bibfnamefont {E.}~\bibnamefont
  {Wigner}},\ }\href@noop {} {\bibfield  {journal} {\bibinfo  {journal} {Z.
  Phys.}\ }\textbf {\bibinfo {volume} {133}},\ \bibinfo {pages} {101} (\bibinfo
  {year} {1952})},\ \bibinfo {note} {{English translation:
  arXiv:1012.4372.}}\BibitemShut {Stop}%
\bibitem [{\citenamefont {Araki}\ and\ \citenamefont
  {Yanase}(1960)}]{PhysRev.120.622}%
  \BibitemOpen
  \bibfield  {author} {\bibinfo {author} {\bibfnamefont {H.}~\bibnamefont
  {Araki}}\ and\ \bibinfo {author} {\bibfnamefont {M.~M.}\ \bibnamefont
  {Yanase}},\ }\href {\doibase 10.1103/PhysRev.120.622} {\bibfield  {journal}
  {\bibinfo  {journal} {Phys. Rev.}\ }\textbf {\bibinfo {volume} {120}},\
  \bibinfo {pages} {622} (\bibinfo {year} {1960})}\BibitemShut {NoStop}%
\bibitem [{\citenamefont {Yanase}(1961)}]{PhysRev.123.666}%
  \BibitemOpen
  \bibfield  {author} {\bibinfo {author} {\bibfnamefont {M.~M.}\ \bibnamefont
  {Yanase}},\ }\href {\doibase 10.1103/PhysRev.123.666} {\bibfield  {journal}
  {\bibinfo  {journal} {Phys. Rev.}\ }\textbf {\bibinfo {volume} {123}},\
  \bibinfo {pages} {666} (\bibinfo {year} {1961})}\BibitemShut {NoStop}%
\bibitem [{\citenamefont {Ozawa}(2002{\natexlab{a}})}]{OzawaWAY}%
  \BibitemOpen
  \bibfield  {author} {\bibinfo {author} {\bibfnamefont {M.}~\bibnamefont
  {Ozawa}},\ }\href {\doibase 10.1103/PhysRevLett.88.050402} {\bibfield
  {journal} {\bibinfo  {journal} {Phys. Rev. Lett.}\ }\textbf {\bibinfo
  {volume} {88}},\ \bibinfo {pages} {050402} (\bibinfo {year}
  {2002}{\natexlab{a}})}\BibitemShut {NoStop}%
\bibitem [{\citenamefont {Korzekwa}(2013)}]{korzekwa2013}%
  \BibitemOpen
  \bibfield  {author} {\bibinfo {author} {\bibfnamefont {K.}~\bibnamefont
  {Korzekwa}},\ }\emph {\bibinfo {title} {Resource theory of asymmetry}},\
  \href {http://kamilkorzekwa.com/theses/mres_thesis.pdf} {Ph.D. thesis}
  (\bibinfo {year} {2013})\BibitemShut {NoStop}%
\bibitem [{\citenamefont {Tajima}\ and\ \citenamefont {Nagaoka}(2019)}]{TN}%
  \BibitemOpen
  \bibfield  {author} {\bibinfo {author} {\bibfnamefont {H.}~\bibnamefont
  {Tajima}}\ and\ \bibinfo {author} {\bibfnamefont {H.}~\bibnamefont
  {Nagaoka}},\ }\href@noop {} {\enquote {\bibinfo {title} {Coherence-variance
  uncertainty relation and coherence cost for quantum measurement under
  conservation laws},}\ }\bibinfo {howpublished} {arXiv:1909.02904} (\bibinfo
  {year} {2019})\BibitemShut {NoStop}%
\bibitem [{\citenamefont {Ozawa}(2002{\natexlab{b}})}]{ozawaWAY_CNOT}%
  \BibitemOpen
  \bibfield  {author} {\bibinfo {author} {\bibfnamefont {M.}~\bibnamefont
  {Ozawa}},\ }\href {\doibase 10.1103/PhysRevLett.89.057902} {\bibfield
  {journal} {\bibinfo  {journal} {Phys. Rev. Lett.}\ }\textbf {\bibinfo
  {volume} {89}},\ \bibinfo {pages} {057902} (\bibinfo {year}
  {2002}{\natexlab{b}})}\BibitemShut {NoStop}%
\bibitem [{\citenamefont {Ozawa}(2003)}]{ozawa2003uncertainty}%
  \BibitemOpen
  \bibfield  {author} {\bibinfo {author} {\bibfnamefont {M.}~\bibnamefont
  {Ozawa}},\ }\href@noop {} {\bibfield  {journal} {\bibinfo  {journal}
  {International Journal of Quantum Information}\ }\textbf {\bibinfo {volume}
  {1}},\ \bibinfo {pages} {569} (\bibinfo {year} {2003})}\BibitemShut {NoStop}%
\bibitem [{\citenamefont {Karasawa}\ and\ \citenamefont
  {Ozawa}(2007)}]{PhysRevA.75.032324}%
  \BibitemOpen
  \bibfield  {author} {\bibinfo {author} {\bibfnamefont {T.}~\bibnamefont
  {Karasawa}}\ and\ \bibinfo {author} {\bibfnamefont {M.}~\bibnamefont
  {Ozawa}},\ }\href {\doibase 10.1103/PhysRevA.75.032324} {\bibfield  {journal}
  {\bibinfo  {journal} {Phys. Rev. A}\ }\textbf {\bibinfo {volume} {75}},\
  \bibinfo {pages} {032324} (\bibinfo {year} {2007})}\BibitemShut {NoStop}%
\bibitem [{\citenamefont {Karasawa}\ \emph {et~al.}(2009)\citenamefont
  {Karasawa}, \citenamefont {Gea-Banacloche},\ and\ \citenamefont
  {Ozawa}}]{Karasawa_2009}%
  \BibitemOpen
  \bibfield  {author} {\bibinfo {author} {\bibfnamefont {T.}~\bibnamefont
  {Karasawa}}, \bibinfo {author} {\bibfnamefont {J.}~\bibnamefont
  {Gea-Banacloche}}, \ and\ \bibinfo {author} {\bibfnamefont {M.}~\bibnamefont
  {Ozawa}},\ }\href {\doibase 10.1088/1751-8113/42/22/225303} {\bibfield
  {journal} {\bibinfo  {journal} {Journal of Physics A: Mathematical and
  Theoretical}\ }\textbf {\bibinfo {volume} {42}},\ \bibinfo {pages} {225303}
  (\bibinfo {year} {2009})}\BibitemShut {NoStop}%
\bibitem [{\citenamefont {Tajima}\ \emph {et~al.}(2018)\citenamefont {Tajima},
  \citenamefont {Shiraishi},\ and\ \citenamefont {Saito}}]{TSS}%
  \BibitemOpen
  \bibfield  {author} {\bibinfo {author} {\bibfnamefont {H.}~\bibnamefont
  {Tajima}}, \bibinfo {author} {\bibfnamefont {N.}~\bibnamefont {Shiraishi}}, \
  and\ \bibinfo {author} {\bibfnamefont {K.}~\bibnamefont {Saito}},\
  }\href@noop {} {\bibfield  {journal} {\bibinfo  {journal} {Phys. Rev. Lett}\
  }\textbf {\bibinfo {volume} {121}},\ \bibinfo {pages} {110403} (\bibinfo
  {year} {2018})}\BibitemShut {NoStop}%
\bibitem [{\citenamefont {Tajima}\ \emph {et~al.}(2020)\citenamefont {Tajima},
  \citenamefont {Shiraishi},\ and\ \citenamefont {Saito}}]{TSS2}%
  \BibitemOpen
  \bibfield  {author} {\bibinfo {author} {\bibfnamefont {H.}~\bibnamefont
  {Tajima}}, \bibinfo {author} {\bibfnamefont {N.}~\bibnamefont {Shiraishi}}, \
  and\ \bibinfo {author} {\bibfnamefont {K.}~\bibnamefont {Saito}},\ }\href
  {\doibase 10.1103/PhysRevResearch.2.043374} {\bibfield  {journal} {\bibinfo
  {journal} {Phys. Rev. Research}\ }\textbf {\bibinfo {volume} {2}},\ \bibinfo
  {pages} {043374} (\bibinfo {year} {2020})}\BibitemShut {NoStop}%
\bibitem [{\citenamefont {Tajima}\ and\ \citenamefont {Saito}(2021)}]{TS}%
  \BibitemOpen
  \bibfield  {author} {\bibinfo {author} {\bibfnamefont {H.}~\bibnamefont
  {Tajima}}\ and\ \bibinfo {author} {\bibfnamefont {K.}~\bibnamefont {Saito}},\
  }\href@noop {} {\enquote {\bibinfo {title} {Universal limitation of quantum
  information recovery: symmetry versus coherence},}\ }\bibinfo {howpublished}
  {arXiv:2103.01876} (\bibinfo {year} {2021})\BibitemShut {NoStop}%
\bibitem [{\citenamefont {Tajima}\ \emph {et~al.}(2022)\citenamefont {Tajima},
  \citenamefont {Takagi},\ and\ \citenamefont {Kuramochi}}]{arxiv.2206.11086}%
  \BibitemOpen
  \bibfield  {author} {\bibinfo {author} {\bibfnamefont {H.}~\bibnamefont
  {Tajima}}, \bibinfo {author} {\bibfnamefont {R.}~\bibnamefont {Takagi}}, \
  and\ \bibinfo {author} {\bibfnamefont {Y.}~\bibnamefont {Kuramochi}},\ }\href
  {\doibase 10.48550/ARXIV.2206.11086} {\enquote {\bibinfo {title} {Universal
  trade-off structure between symmetry, irreversibility, and quantum coherence
  in quantum processes},}\ }\bibinfo {howpublished} {arXiv:2206.11086}
  (\bibinfo {year} {2022})\BibitemShut {NoStop}%
\bibitem [{\citenamefont {Stein}\ and\ \citenamefont
  {Shimony}(1971)}]{steinshimony1971}%
  \BibitemOpen
  \bibfield  {author} {\bibinfo {author} {\bibfnamefont {H.}~\bibnamefont
  {Stein}}\ and\ \bibinfo {author} {\bibfnamefont {A.}~\bibnamefont
  {Shimony}},\ }in\ \href@noop {} {\emph {\bibinfo {booktitle} {{Foundations of
  Quantum Mechanics edited by B. d’Espagnat}}}}\ (\bibinfo  {publisher}
  {Academic Press, New York},\ \bibinfo {year} {1971})\ pp.\ \bibinfo {pages}
  {56--76}\BibitemShut {NoStop}%
\bibitem [{\citenamefont {Ghirardi}\ \emph {et~al.}(1983)\citenamefont
  {Ghirardi}, \citenamefont {Rimini},\ and\ \citenamefont
  {Weber}}]{doi:10.1063/1.525627}%
  \BibitemOpen
  \bibfield  {author} {\bibinfo {author} {\bibfnamefont {G.~C.}\ \bibnamefont
  {Ghirardi}}, \bibinfo {author} {\bibfnamefont {A.}~\bibnamefont {Rimini}}, \
  and\ \bibinfo {author} {\bibfnamefont {T.}~\bibnamefont {Weber}},\ }\href
  {\doibase 10.1063/1.525627} {\bibfield  {journal} {\bibinfo  {journal} {J.
  Math. Phys.}\ }\textbf {\bibinfo {volume} {24}},\ \bibinfo {pages} {2454}
  (\bibinfo {year} {1983})},\ \Eprint
  {http://arxiv.org/abs/https://doi.org/10.1063/1.525627}
  {https://doi.org/10.1063/1.525627} \BibitemShut {NoStop}%
\bibitem [{\citenamefont {Ozawa}(1991)}]{PhysRevLett.67.1956}%
  \BibitemOpen
  \bibfield  {author} {\bibinfo {author} {\bibfnamefont {M.}~\bibnamefont
  {Ozawa}},\ }\href {\doibase 10.1103/PhysRevLett.67.1956} {\bibfield
  {journal} {\bibinfo  {journal} {Phys. Rev. Lett.}\ }\textbf {\bibinfo
  {volume} {67}},\ \bibinfo {pages} {1956} (\bibinfo {year}
  {1991})}\BibitemShut {NoStop}%
\bibitem [{\citenamefont {Ozawa}(1993)}]{ozawa1993way}%
  \BibitemOpen
  {\color{black}
  \bibfield  {author} {\bibinfo {author} {\bibfnamefont {M.}~\bibnamefont
  {Ozawa}},\ }in\ \href@noop {} {\emph {\bibinfo {booktitle} {Classical and
  Quantum Systems: Foundations and Symmetries}}},\ \bibinfo {editor} {edited
  by\ \bibinfo {editor} {\bibfnamefont {H.~D.}\ \bibnamefont {Doebner}},
  \bibinfo {editor} {\bibfnamefont {W.}~\bibnamefont {Scherer}}, \ and\
  \bibinfo {editor} {\bibfnamefont {F.}~\bibnamefont {Schroeck}, \bibfnamefont
  {Jr}}}\ (\bibinfo  {publisher} {World Scientific},\ \bibinfo {year} {1993})\
  pp.\ \bibinfo {pages} {224--227}}\BibitemShut {NoStop}%
\bibitem [{\citenamefont {Busch}\ and\ \citenamefont
  {Loveridge}(2011)}]{PhysRevLett.106.110406}%
  \BibitemOpen
  \bibfield  {author} {\bibinfo {author} {\bibfnamefont {P.}~\bibnamefont
  {Busch}}\ and\ \bibinfo {author} {\bibfnamefont {L.}~\bibnamefont
  {Loveridge}},\ }\href {\doibase 10.1103/PhysRevLett.106.110406} {\bibfield
  {journal} {\bibinfo  {journal} {Phys. Rev. Lett.}\ }\textbf {\bibinfo
  {volume} {106}},\ \bibinfo {pages} {110406} (\bibinfo {year}
  {2011})}\BibitemShut {NoStop}%
\bibitem [{\citenamefont {Loveridge}(2020)}]{Loveridge_2020}%
  \BibitemOpen
  \bibfield  {author} {\bibinfo {author} {\bibfnamefont {L.}~\bibnamefont
  {Loveridge}},\ }\href {\doibase 10.1088/1742-6596/1638/1/012009} {\bibfield
  {journal} {\bibinfo  {journal} {J. Phys. Conf. Ser.}\ }\textbf {\bibinfo
  {volume} {1638}},\ \bibinfo {pages} {012009} (\bibinfo {year}
  {2020})}\BibitemShut {NoStop}%
\bibitem [{Note1()}]{Note1}%
  \BibitemOpen
  \bibinfo {note} {\protect \leavevmode {\protect \color {black}Here, we only
  explain the case of $L_S=L_{S'}$. Our theorem also covers the case of
  $L_S\protect \ne L_{S'}$. For the detail, see Theorem \ref
  {thm:main2}.}}\BibitemShut {Stop}%
\bibitem [{\citenamefont {Stinespring}(1955)}]{1955stinespring}%
  \BibitemOpen
  \bibfield  {author} {\bibinfo {author} {\bibfnamefont {W.~F.}\ \bibnamefont
  {Stinespring}},\ }\href {http://www.jstor.org/stable/2032342} {\bibfield
  {journal} {\bibinfo  {journal} {Proc. Amer. Math. Soc.}\ }\textbf {\bibinfo
  {volume} {6}},\ \bibinfo {pages} {211} (\bibinfo {year} {1955})}\BibitemShut
  {NoStop}%
\bibitem [{\citenamefont {Paulsen}(2003)}]{paulsen_2003}%
  \BibitemOpen
  \bibfield  {author} {\bibinfo {author} {\bibfnamefont {V.}~\bibnamefont
  {Paulsen}},\ }\href {\doibase 10.1017/CBO9780511546631} {\emph {\bibinfo
  {title} {Completely Bounded Maps and Operator Algebras}}},\ Cambridge Studies
  in Advanced Mathematics\ (\bibinfo  {publisher} {Cambridge University
  Press},\ \bibinfo {year} {2003})\BibitemShut {NoStop}%
\bibitem [{\citenamefont {Davies}(1976)}]{davies1976quantum}%
  \BibitemOpen
  \bibfield  {author} {\bibinfo {author} {\bibfnamefont {E.~B.}\ \bibnamefont
  {Davies}},\ }\href@noop {} {\emph {\bibinfo {title} {Quantum theory of open
  systems}}}\ (\bibinfo  {publisher} {IMA},\ \bibinfo {year}
  {1976})\BibitemShut {NoStop}%
\bibitem [{\citenamefont {Holevo}(2001)}]{holevo2001statistical}%
  \BibitemOpen
  \bibfield  {author} {\bibinfo {author} {\bibfnamefont {A.~S.}\ \bibnamefont
  {Holevo}},\ }\href@noop {} {\emph {\bibinfo {title} {Statistical Structure of
  Quantum Theory}}}\ (\bibinfo  {publisher} {Springer},\ \bibinfo {year}
  {2001})\BibitemShut {NoStop}%
\bibitem [{\citenamefont {Reed}\ and\ \citenamefont
  {Simon}(1972)}]{reed1972methods}%
  \BibitemOpen
  \bibfield  {author} {\bibinfo {author} {\bibfnamefont {M.}~\bibnamefont
  {Reed}}\ and\ \bibinfo {author} {\bibfnamefont {B.}~\bibnamefont {Simon}},\
  }\href@noop {} {\emph {\bibinfo {title} {Methods of modern mathematical
  physics}}},\ Vol.~\bibinfo {volume} {1}\ (\bibinfo  {publisher} {Elsevier},\
  \bibinfo {year} {1972})\BibitemShut {NoStop}%
\bibitem [{\citenamefont {Prugovecki}(1982)}]{prugovecki1982quantum}%
  \BibitemOpen
  \bibfield  {author} {\bibinfo {author} {\bibfnamefont {E.}~\bibnamefont
  {Prugovecki}},\ }\href@noop {} {\emph {\bibinfo {title} {{Quantum mechanics
  in Hilbert space}}}}\ (\bibinfo  {publisher} {Academic Press},\ \bibinfo
  {year} {1982})\BibitemShut {NoStop}%
\bibitem [{\citenamefont {Choi}(1974)}]{choi1974}%
  \BibitemOpen
  \bibfield  {author} {\bibinfo {author} {\bibfnamefont {M.-D.}\ \bibnamefont
  {Choi}},\ }\href {http://projecteuclid.org/euclid.ijm/1256051007} {\bibfield
  {journal} {\bibinfo  {journal} {Illinois J. Math.}\ }\textbf {\bibinfo
  {volume} {18}},\ \bibinfo {pages} {565} (\bibinfo {year} {1974})}\BibitemShut
  {NoStop}%
\bibitem [{\citenamefont {Mohammady}\ \emph {et~al.}(2021)\citenamefont
  {Mohammady}, \citenamefont {Miyadera},\ and\ \citenamefont
  {Loveridge}}]{mohammady2022}%
  \BibitemOpen
  \bibfield  {author} {\bibinfo {author} {\bibfnamefont {M.~H.}\ \bibnamefont
  {Mohammady}}, \bibinfo {author} {\bibfnamefont {T.}~\bibnamefont {Miyadera}},
  \ and\ \bibinfo {author} {\bibfnamefont {L.}~\bibnamefont {Loveridge}},\
  }\href {\doibase 10.48550/ARXIV.2110.11705} {\enquote {\bibinfo {title}
  {Measurement disturbance and conservation laws in quantum mechanics},}\ }
  (\bibinfo {year} {2021}),\ \Eprint {http://arxiv.org/abs/arxiv:2110.11705}
  {arxiv:2110.11705} \BibitemShut {NoStop}%
\bibitem [{\citenamefont {Pontryagin}(1986)}]{pontryagin1986topological}%
  \BibitemOpen
  \bibfield  {author} {\bibinfo {author} {\bibfnamefont {L.~S.}\ \bibnamefont
  {Pontryagin}},\ }\href@noop {} {\emph {\bibinfo {title} {Topological
  groups}}},\ \bibinfo {edition} {3rd}\ ed.\ (\bibinfo  {publisher}
  {Routledge},\ \bibinfo {year} {1986})\ \bibinfo {note} {trans. from Russian
  by Arlen Brown and P.S.V. Naidu}\BibitemShut {NoStop}%
\bibitem [{\citenamefont {Higgins}(1975)}]{higgins_1975}%
  \BibitemOpen
  \bibfield  {author} {\bibinfo {author} {\bibfnamefont {P.~J.}\ \bibnamefont
  {Higgins}},\ }\href {\doibase 10.1017/CBO9781107359918} {\emph {\bibinfo
  {title} {An Introduction to Topological Groups}}},\ London Mathematical
  Society Lecture Note Series\ (\bibinfo  {publisher} {Cambridge University
  Press},\ \bibinfo {year} {1975})\BibitemShut {NoStop}%
\bibitem [{\citenamefont {Hirano}\ \emph {et~al.}(2017)\citenamefont {Hirano},
  \citenamefont {Ichikawa}, \citenamefont {Matsubara}, \citenamefont {Ono},
  \citenamefont {Oguri}, \citenamefont {Namiki}, \citenamefont {Kasai},
  \citenamefont {Matsumoto},\ and\ \citenamefont {Tsurumaru}}]{Hirano_2017}%
  \BibitemOpen
  {\color{black}
  \bibfield  {author} {\bibinfo {author} {\bibfnamefont {T.}~\bibnamefont
  {Hirano}}, \bibinfo {author} {\bibfnamefont {T.}~\bibnamefont {Ichikawa}},
  \bibinfo {author} {\bibfnamefont {T.}~\bibnamefont {Matsubara}}, \bibinfo
  {author} {\bibfnamefont {M.}~\bibnamefont {Ono}}, \bibinfo {author}
  {\bibfnamefont {Y.}~\bibnamefont {Oguri}}, \bibinfo {author} {\bibfnamefont
  {R.}~\bibnamefont {Namiki}}, \bibinfo {author} {\bibfnamefont
  {K.}~\bibnamefont {Kasai}}, \bibinfo {author} {\bibfnamefont
  {R.}~\bibnamefont {Matsumoto}}, \ and\ \bibinfo {author} {\bibfnamefont
  {T.}~\bibnamefont {Tsurumaru}},\ }\href {\doibase 10.1088/2058-9565/aa7230}
  {\bibfield  {journal} {\bibinfo  {journal} {Quantum Sci. Technol.}\ }\textbf
  {\bibinfo {volume} {2}},\ \bibinfo {pages} {024010} (\bibinfo {year}
  {2017})}}\BibitemShut {NoStop}%
\bibitem [{\citenamefont {Furusawa}\ \emph {et~al.}(1998)\citenamefont
  {Furusawa}, \citenamefont {Sørensen}, \citenamefont {Braunstein},
  \citenamefont {Fuchs}, \citenamefont {Kimble},\ and\ \citenamefont
  {Polzik}}]{furusawa1998}%
  \BibitemOpen
  {\color{black}
  \bibfield  {author} {\bibinfo {author} {\bibfnamefont {A.}~\bibnamefont
  {Furusawa}}, \bibinfo {author} {\bibfnamefont {J.~L.}\ \bibnamefont
  {Sørensen}}, \bibinfo {author} {\bibfnamefont {S.~L.}\ \bibnamefont
  {Braunstein}}, \bibinfo {author} {\bibfnamefont {C.~A.}\ \bibnamefont
  {Fuchs}}, \bibinfo {author} {\bibfnamefont {H.~J.}\ \bibnamefont {Kimble}}, \
  and\ \bibinfo {author} {\bibfnamefont {E.~S.}\ \bibnamefont {Polzik}},\
  }\href {\doibase 10.1126/science.282.5389.706} {\bibfield  {journal}
  {\bibinfo  {journal} {Science}\ }\textbf {\bibinfo {volume} {282}},\ \bibinfo
  {pages} {706} (\bibinfo {year} {1998})}}\BibitemShut {NoStop}%
\bibitem [{\citenamefont {Yuen}\ and\ \citenamefont
  {Shapiro}(1978)}]{yuen1978quantum}%
  \BibitemOpen
  \bibfield  {author} {\bibinfo {author} {\bibfnamefont {H.~P.}\ \bibnamefont
  {Yuen}}\ and\ \bibinfo {author} {\bibfnamefont {J.~H.}\ \bibnamefont
  {Shapiro}},\ }in\ \href {\doibase 10.1007/978-1-4757-0665-9_75} {\emph
  {\bibinfo {booktitle} {Coherence and Quantum Optics IV}}}\ (\bibinfo
  {publisher} {Springer},\ \bibinfo {year} {1978})\ pp.\ \bibinfo {pages}
  {719--727}\BibitemShut {NoStop}%
\bibitem [{\citenamefont {Tyc}\ and\ \citenamefont {Sanders}(2004)}]{Tyc_2004}%
  \BibitemOpen
  \bibfield  {author} {\bibinfo {author} {\bibfnamefont {T.}~\bibnamefont
  {Tyc}}\ and\ \bibinfo {author} {\bibfnamefont {B.~C.}\ \bibnamefont
  {Sanders}},\ }\href {\doibase 10.1088/0305-4470/37/29/010} {\bibfield
  {journal} {\bibinfo  {journal} {J. Phys. A: Math. Gen.}\ }\textbf {\bibinfo
  {volume} {37}},\ \bibinfo {pages} {7341} (\bibinfo {year}
  {2004})}\BibitemShut {NoStop}%
\bibitem [{\citenamefont
  {Ozawa}(1984)}]{:/content/aip/journal/jmp/25/1/10.1063/1.526000}%
  \BibitemOpen
  \bibfield  {author} {\bibinfo {author} {\bibfnamefont {M.}~\bibnamefont
  {Ozawa}},\ }\href {\doibase 10.1063/1.526000} {\bibfield  {journal} {\bibinfo
   {journal} {J. Math. Phys.}\ }\textbf {\bibinfo {volume} {25}},\ \bibinfo
  {pages} {79} (\bibinfo {year} {1984})}\BibitemShut {NoStop}%
\bibitem [{Note2()}]{Note2}%
  \BibitemOpen
  \bibinfo {note} {\protect \leavevmode {\protect 
  See Supplemental Material for a brief review of the WAY-type trade-off relations between the implementation error and the required resource in quantum processes. Supplemental Material includes the references~\cite{skew_resource,Takagi_skew,YT,Marvian_distillation,Hansen,%
kudo_fisher_2022,min_V_Yu,min_V_Petz,TT,Yuxiang1,Yuxiang2,liu_quantum_2022} in addition to the references \cite{PhysRev.123.666,OzawaWAY,ozawaWAY_CNOT,ozawa2003uncertainty,PhysRevA.75.032324,Karasawa_2009,TSS, korzekwa2013,TN,TSS2,TS,arxiv.2206.11086,paulsen_2003,reed1972methods,prugovecki1982quantum,choi1974,pontryagin1986topological,higgins_1975} cited in the main text.  }}\BibitemShut {Stop}%
\bibitem [{\citenamefont {Zhang}\ \emph {et~al.}(2017)\citenamefont {Zhang},
  \citenamefont {Yadin}, \citenamefont {Hou}, \citenamefont {Cao},
  \citenamefont {Liu}, \citenamefont {Huang}, \citenamefont {Maity},
  \citenamefont {Vedral}, \citenamefont {Li}, \citenamefont {Guo} \emph
  {et~al.}}]{skew_resource}%
  \BibitemOpen
  \bibfield  {author} {\bibinfo {author} {\bibfnamefont {C.}~\bibnamefont
  {Zhang}}, \bibinfo {author} {\bibfnamefont {B.}~\bibnamefont {Yadin}},
  \bibinfo {author} {\bibfnamefont {Z.-B.}\ \bibnamefont {Hou}}, \bibinfo
  {author} {\bibfnamefont {H.}~\bibnamefont {Cao}}, \bibinfo {author}
  {\bibfnamefont {B.-H.}\ \bibnamefont {Liu}}, \bibinfo {author} {\bibfnamefont
  {Y.-F.}\ \bibnamefont {Huang}}, \bibinfo {author} {\bibfnamefont
  {R.}~\bibnamefont {Maity}}, \bibinfo {author} {\bibfnamefont
  {V.}~\bibnamefont {Vedral}}, \bibinfo {author} {\bibfnamefont {C.-F.}\
  \bibnamefont {Li}}, \bibinfo {author} {\bibfnamefont {G.-C.}\ \bibnamefont
  {Guo}},  \emph {et~al.},\ }\href {https://doi.org/10.1103/PhysRevA.96.042327}
  {\bibfield  {journal} {\bibinfo  {journal} {Phys. Rev. A}\ }\textbf {\bibinfo
  {volume} {96}},\ \bibinfo {pages} {042327} (\bibinfo {year}
  {2017})}\BibitemShut {NoStop}%
\bibitem [{\citenamefont {Takagi}(2019)}]{Takagi_skew}%
  \BibitemOpen
  \bibfield  {author} {\bibinfo {author} {\bibfnamefont {R.}~\bibnamefont
  {Takagi}},\ }\href {https://doi.org/10.1038/s41598-019-50279-w} {\bibfield
  {journal} {\bibinfo  {journal} {Sci. Rep.}\ }\textbf {\bibinfo {volume}
  {9}},\ \bibinfo {pages} {14562} (\bibinfo {year} {2019})}\BibitemShut
  {NoStop}%
\bibitem [{\citenamefont {Yamaguchi}\ and\ \citenamefont {Tajima}(2022)}]{YT}%
  \BibitemOpen
  \bibfield  {author} {\bibinfo {author} {\bibfnamefont {K.}~\bibnamefont
  {Yamaguchi}}\ and\ \bibinfo {author} {\bibfnamefont {H.}~\bibnamefont
  {Tajima}},\ }\href {https://doi.org/10.48550/arXiv.2204.08439}
   {\enquote
  {\bibinfo {title} {Beyond i.i.d. in the resource theory of asymmetry: An
  information-spectrum approach for quantum {Fisher} information},}\ }
  \bibinfo
  {howpublished} {arXiv:2204.08439} (\bibinfo {year} {2022})\BibitemShut
  {NoStop}%
\bibitem [{\citenamefont {Marvian}(2020)}]{Marvian_distillation}%
  \BibitemOpen
  \bibfield  {author} {\bibinfo {author} {\bibfnamefont {I.}~\bibnamefont
  {Marvian}},\ }\href {https://doi.org/10.1038/s41467-019-13846-3} {\bibfield
  {journal} {\bibinfo  {journal} {Nat. Commun.}\ }\textbf {\bibinfo {volume}
  {11}},\ \bibinfo {pages} {25} (\bibinfo {year} {2020})}\BibitemShut {NoStop}%
\bibitem [{\citenamefont {Hansen}(2008)}]{Hansen}%
  \BibitemOpen
  \bibfield  {author} {\bibinfo {author} {\bibfnamefont {F.}~\bibnamefont
  {Hansen}},\ }\href {\doibase 10.1073/pnas.0803323105} {\bibfield  {journal}
  {\bibinfo  {journal} {PNAS}\ }\textbf {\bibinfo {volume} {105}},\ \bibinfo
  {pages} {9909} (\bibinfo {year} {2008})} \BibitemShut {NoStop}%
\bibitem [{\citenamefont {Kudo}\ and\ \citenamefont
  {Tajima}(2022)}]{kudo_fisher_2022}%
  \BibitemOpen
  \bibfield  {author} {\bibinfo {author} {\bibfnamefont {D.}~\bibnamefont
  {Kudo}}\ and\ \bibinfo {author} {\bibfnamefont {H.}~\bibnamefont {Tajima}},\
  }\href@noop {} 
%  {\enquote {\bibinfo {title} {{Fisher} information matrix as a
%  resource measure in resource theory of asymmetry with general connected lie
%  group symmetry},}\ } 
  (\bibinfo {year} {2022}),\ \Eprint
  {http://arxiv.org/abs/2205.03245} {arXiv:2205.03245 [quant-ph]} \BibitemShut
  {NoStop}%
\bibitem [{\citenamefont {Yu}(2013)}]{min_V_Yu}%
  \BibitemOpen
  \bibfield  {author} {\bibinfo {author} {\bibfnamefont {S.}~\bibnamefont
  {Yu}},\ }\href {\doibase 10.48550/ARXIV.1302.5311} {\enquote {\bibinfo
  {title} {Quantum Fisher information as the convex roof of variance},}\ }
  (\bibinfo {year} {2013})\BibitemShut {NoStop}%
\bibitem [{\citenamefont {T{\'o}th}\ and\ \citenamefont
  {Petz}(2013)}]{min_V_Petz}%
  \BibitemOpen
  \bibfield  {author} {\bibinfo {author} {\bibfnamefont {G.}~\bibnamefont
  {T{\'o}th}}\ and\ \bibinfo {author} {\bibfnamefont {D.}~\bibnamefont
  {Petz}},\ }\href {https://doi.org/10.1103/PhysRevA.87.032324} {\bibfield
  {journal} {\bibinfo  {journal} {Phys. Rev. A}\ }\textbf {\bibinfo {volume}
  {87}},\ \bibinfo {pages} {032324} (\bibinfo {year} {2013})}\BibitemShut
  {NoStop}%
\bibitem [{\citenamefont {Takagi}\ and\ \citenamefont {Tajima}(2020)}]{TT}%
  \BibitemOpen
  \bibfield  {author} {\bibinfo {author} {\bibfnamefont {R.}~\bibnamefont
  {Takagi}}\ and\ \bibinfo {author} {\bibfnamefont {H.}~\bibnamefont
  {Tajima}},\ }\href {\doibase 10.1103/PhysRevA.101.022315} {\bibfield
  {journal} {\bibinfo  {journal} {Phys. Rev. A}\ }\textbf {\bibinfo {volume}
  {101}},\ \bibinfo {pages} {022315} (\bibinfo {year} {2020})}\BibitemShut
  {NoStop}%
\bibitem [{\citenamefont {Chiribella}\ \emph {et~al.}(2021)\citenamefont
  {Chiribella}, \citenamefont {Yang},\ and\ \citenamefont {Renner}}]{Yuxiang1}%
  \BibitemOpen
  \bibfield  {author} {\bibinfo {author} {\bibfnamefont {G.}~\bibnamefont
  {Chiribella}}, \bibinfo {author} {\bibfnamefont {Y.}~\bibnamefont {Yang}}, \
  and\ \bibinfo {author} {\bibfnamefont {R.}~\bibnamefont {Renner}},\ }\href
  {\doibase 10.1103/PhysRevX.11.021014} {\bibfield  {journal} {\bibinfo
  {journal} {Phys. Rev. X}\ }\textbf {\bibinfo {volume} {11}},\ \bibinfo
  {pages} {021014} (\bibinfo {year} {2021})}\BibitemShut {NoStop}%
\bibitem [{\citenamefont {Yang}\ \emph {et~al.}(2022)\citenamefont {Yang},
  \citenamefont {Renner},\ and\ \citenamefont {Chiribella}}]{Yuxiang2}%
  \BibitemOpen
  \bibfield  {author} {\bibinfo {author} {\bibfnamefont {Y.}~\bibnamefont
  {Yang}}, \bibinfo {author} {\bibfnamefont {R.}~\bibnamefont {Renner}}, \ and\
  \bibinfo {author} {\bibfnamefont {G.}~\bibnamefont {Chiribella}},\
  }\href@noop {} {\bibfield  {journal} {\bibinfo  {journal} {accepted
  manuscript for J. Phys. A: Math. Theor.}\ } (\bibinfo {year}
  {2022})}\BibitemShut {NoStop}%
\bibitem [{\citenamefont {Liu}\ and\ \citenamefont
  {Zhou}(2022)}]{liu_quantum_2022}%
  \BibitemOpen
  \bibfield  {author} {\bibinfo {author} {\bibfnamefont {Z.-W.}\ \bibnamefont
  {Liu}}\ and\ \bibinfo {author} {\bibfnamefont {S.}~\bibnamefont {Zhou}},\
  }\href {http://arxiv.org/abs/2111.06360} {\bibfield  {journal} {\bibinfo
  {journal} {arXiv:2111.06360 [quant-ph]}\ } (\bibinfo {year}
  {2022})}\BibitemShut {NoStop}%
\end{thebibliography}

\clearpage

%\nolinenumbers

\begin{widetext}
\begin{center}
{\large \bf Supplemental Material for \protect \\ 
``Wigner-Araki-Yanase theorem for continuous and unbounded conserved observables''}\\
\vspace*{0.3cm}
Yui Kuramochi$^{1}$ and Hiroyasu Tajima$^{2}$\\
\vspace*{0.1cm}
$^{1}${\small \em Department of Physics, Kyushu University, 744 Motooka, Nishi-ku, Fukuoka, Japan}
\\
$^{2}${\small \em Graduate School of Informatics and Engineering,
The University of Electro-Communications,
1-5-1 Chofugaoka, Chofu, Tokyo 182-8585, Japan}
\end{center}

\setcounter{equation}{0}
\setcounter{lemm}{0}
\setcounter{thm}{0}
\setcounter{page}{1}
\renewcommand{\theequation}{S.\arabic{equation}}
\renewcommand{\thethm}{S.\arabic{thm}}
\renewcommand{\thedefi}{S.\arabic{defi}}

This supplemental material is organized as follows.
We first state WAY-type theorems for general continuous symmetries (Theorems~\ref{thm:main1g} and \ref{thm:main2g}) and show that Theorems~\ref{thm:main1} and \ref{thm:main2} in the main part follow from these general theorems as corollaries.
We then prove Theorems~\ref{thm:main1g} and \ref{thm:main2g}.
% in Section~\ref{sec:proof1}.
We also give a brief review of the WAY-type trade-off relations between the implementation error and the required resource in various quantum processes.

\section{WAY-theorems for general continuous symmetries} \label{sec:maing}
To state the main theorems, we first introduce the notion of the strongly continuous unitary representation of a topological group.

\begin{defi}[Continuous unitary representation of a topological group]
\begin{enumerate}[(i)]
\item
$G$ is called a \textit{topological group} \cite{pontryagin1986topological,higgins_1975} if $G$ is a topological space that is also a group such that the group multiplication
\begin{equation}
	G \times G \ni (g,h) \mapsto gh \in G
\end{equation}
and the inverse
\begin{equation}
	G \ni g \mapsto g^{-1} \in G
\end{equation}
are both continuous maps.
A topological group $G$ is said to be \textit{connected} if $G$ is connected as a topological space, i.e.\ $G$ has no proper closed and open subset other than the empty set.
\item
For a group $G$, $\gin{U_g}$ is called a \textit{unitary representation} of $G$ acting on a Hilbert space $\cH_A$ if $U_g$ is a unitary operator on $\cH_A$ for every $g \in G$ and 
\begin{gather}
	U_e = \unit_A ,\\
	U_{gh^{-1}} = U_g U_h^\dag \quad (g,h \in G),
\end{gather}
i.e.\ $G\ni g \mapsto U_g $ is a group homomorphism.
Here, $e\in G$ denotes the unit element of $G$.
\item
Let $\cH_A$ be a Hilbert space.
The \textit{strong operator topology} (SOT) on $\BH{A}$ is the locally convex topology induced by the semi-norms $\{ p_\xi  \}_{\xi \in \cH_A}$, where for every $\xi \in \cH_A$
\begin{equation}
	p_\xi (a) := \|a\xi \| \quad (a\in \BH{A}).
\end{equation}
In terms of nets (Moore-Smith sequences), the SOT is characterized as follows: a net $(a_i)_{i\in I}$ in $\BH{A}$ converges to $a\in \BH{A}$ in the SOT if and only if $\| (a_i -a)\xi \| \to 0$ for every $\xi \in \cH_A$.
\item
A unitary representation $\gin{U_g}$ of a topological group $G$ acting on $\cH_A$ is said to be \textit{strongly continuous} if $G \ni g \mapsto U_g \in \BH{A}$ is continuous respectively in the topology on $G$ and in the SOT on $\BH{A}$.
\end{enumerate}
\end{defi}

For example, the group
\begin{equation}
	\Uone := \{c \in \cmplx : |c|=1 \} ,
\end{equation}
equipped with the Euclidean topology and the ordinary multiplication and inverse 
\begin{equation}
	\Uone \times \Uone \ni (c_1,c_2) \mapsto c_1 c_2 \in \Uone, \quad
	\Uone \ni c \mapsto c^{-1} = \overline{c} \in \Uone
\end{equation}
is a topological group.
Another example of a topological group is the set $\realn$ of real numbers equipped with the Euclidean topology and the additive group operations
\begin{equation}
	\realn \times \realn \ni (s,t) \mapsto s+t \in \realn,
	\quad 
	\realn \ni t \mapsto -t \in \realn .
\end{equation}
For every self-adjoint operator $L_A$ on a Hilbert space $\cH_A$, $(e^{itL_A})_{t\in \realn}$
is a strongly continuous unitary representation of $\realn$.
Conversely, according to Stone\rq{}s theorem (e.g.\ \cite{prugovecki1982quantum}, Theorem~IV.6.1), for every strongly continuous unitary representation $(U_t)_{t\in \realn}$ of $\realn$ acting on a Hilbert space $\cH_A$, there exists a unique self-adjoint operator $L_A$ on $\cH_A$ such that $U_t = e^{itL_A}$ $(t \in \realn)$.

We now state the main theorems of this supplemental material:

\begin{thm}[WAY theorem for projective measurements and group symmetries] \label{thm:main1g}
Let 
$
\mathbb{M} = (\cH_P, \cH_{S^\prime} , \cH_{P^\prime} ,\rho_P ,\Uspsp , (\Omega, \Sigma, \sfF_\Ppr))
$
be a measurement model that implements a PVM $(\Omega ,\Sigma, \sfE_S)$ on $\cH_S$, let $G$ be a connected topological group, and let $\gin{U_g^S}$, $\gin{U_g^P}$, $\gin{U_g^\Spr}$, and $\gin{U_g^\Ppr}$ be unitary representations of $G$ that act respectively on $\cH_S$, $\cH_P$, $\cH_\Spr$, and $\cH_\Ppr$.
Suppose that the following conditions hold:
\begin{itemize}
\item 
$\gin{U^P_g}$ is a strongly continuous unitary representation;
\item ($G$-invariance condition)
\begin{equation}
	\Uspsp^\dag (U_g^\Spr \otimes U_g^\Ppr) \Uspsp = U_g^S \otimes U_g^P 
	\quad (g\in G);
	\label{eq:Ginv}
\end{equation}
\item (the generalized Yanase condition) $\sfF_\Ppr (X)$ commutes with $U_g^\Ppr$ for every $X \in \Sigma$ and every $g\in G$.
\end{itemize}
Then $\sfE_S(X)$ commutes with $U_g^S$ for every $X\in \Sigma$ and every $g\in G$.
\end{thm}

\begin{thm}[WAY theorem for unitary channels and group symmetries] \label{thm:main2g}
Let $U_{S\to \Spr} \colon \cH_S \to \cH_\Spr$ be a unitary operator between Hilbert spaces $\cH_S$ and $\cH_\Spr$, let $\mathcal{U}_{S\to \Spr} \colon \TrH{S} \to \TrH{\Spr}$ be the unitary channel defined by $\mathcal{U}_{S\to \Spr} (\rho_S) := U_{S\to \Spr} \rho_S U_{S\to \Spr}^\dag$ $(\rho_S \in \TrH{S})$, let $\mathbb{M} = (\cH_P, \cH_{S^\prime} , \cH_{P^\prime} , \rho_P,\Uspsp)$ be a system-environment model that implements $\mathcal{U}_{S\to \Spr}$, and let $G$ be a connected topological group.
Suppose that there are strongly continuous unitary representations $\gin{U_g^S}$, $\gin{U_g^P}$, $\gin{U_g^\Spr}$, and $\gin{U_g^\Ppr}$ of $G$ that act respectively on $\cH_S$, $\cH_P$, $\cH_\Spr$, and $\cH_\Ppr$ and satisfy the $G$-invariance condition~\eqref{eq:Ginv}.
Then there exists a $1$-dimensional continuous unitary representation $G\ni g \mapsto c_g \in \Uone$ such that
\begin{equation}
	U_{S\to \Spr}^\dag U_g^\Spr U_{S\to \Spr} = c_g U^S_g  \quad (g\in G).
	\label{eq:main2g-1}
\end{equation}
\end{thm}

Now we prove Theorems~\ref{thm:main1} and \ref{thm:main2} based on Theorems~\ref{thm:main1g} and \ref{thm:main2g}.

\begin{proof}[Proof of Theorem~\ref{thm:main1}.]
We apply Theorem~\ref{thm:main1g} by putting $G = \realn$ and $U_t^\alpha  = e^{itL_\alpha}$ ($t\in \realn$; $\alpha = S, P, \Spr, \Ppr$).
Then we can easily see that the $G$-invariance and the generalized Yanase conditions in this case follow respectively from the conservation law~\eqref{eq:cons} and the Yanase condition.
Thus, Theorem~\ref{thm:main1g} implies that $\sfE_S(X)$ commutes with $e^{itL_S}$ for all $t\in \realn$ and all $X\in \Sigma$.
But from \cite{reed1972methods} (Theorem~VIII.13), this implies the commutativity of $\sfE_S (X)$ and $L_S$ for all $X\in \Sigma$.
\end{proof}

\begin{proof}[Proof of Theorem~\ref{thm:main2}.]
Similarly we apply Theorem~\ref{thm:main2g} by putting $G = \realn$ and $U_t^\alpha  = e^{itL_\alpha}$ ($t\in \realn$; $\alpha = S, P, \Spr, \Ppr$).
Then it follows that there exists a $1$-dimensional continuous unitary representation $\realn \ni t \mapsto c_t \in \Uone$ such that 
\begin{equation}
	 e^{it U_{S\to \Spr}^\dag L_\Spr U_{S\to \Spr}} 
	 =U_{S\to \Spr}^\dag e^{itL_\Spr} U_{S\to \Spr}
	 =c_t e^{itL_S}
	 \quad (t\in \realn).
	\label{eq:m2-p1}
\end{equation}
From Stone\rq{}s theorem, there exists a real number $\gamma$ such that $c_t = e^{it\gamma }$.
Thus, Eq.~\eqref{eq:m2-p1} implies
\begin{equation}
	e^{it U_{S\to \Spr}^\dag L_\Spr U_{S\to \Spr}} 
	= e^{it (L_S + \gamma \unit_S)} 
	\quad (t\in \realn).
\end{equation}
From the uniqueness part of Stone\rq{}s theorem, this implies Eq.~\eqref{eq:main2-1}.
The latter part of Theorem~\ref{thm:main2} has already been proved in the main part.
\end{proof}

%In the rest of this supplemental material we prove Theorems~\ref{thm:main1g} and \ref{thm:main2g}.

\section{Proofs of Theorems~\ref{thm:main1g} and \ref{thm:main2g}} \label{sec:proof1}
In this section we prove Theorems~\ref{thm:main1g} and \ref{thm:main2g}.

The proofs are based on the following notion of the multiplicative domain of a unital CP map.
Let $\Lambda \colon \BH{A} \to \BH{B}$ be a CP map that is also unital, i.e.\ $\Lambda (\unit_A) = \unit_B$.
We define the \textit{multiplicative domain} of $\Lambda$ by
\begin{equation}
	\mathcal{M}_\Lambda := \{ a \in \BH{A} : \Lambda(a^\dag a) = \Lambda (a^\dag) \Lambda (a) , \, \Lambda(a a^\dag ) = \Lambda (a) \Lambda (a^\dag)  \}.
\end{equation}
Then we have 

\begin{lemm}[\cite{choi1974}; \cite{paulsen_2003}, Proposition~3.3 and Theorem~3.18] \label{lemm:md}
Let $\Lambda \colon \BH{A} \to \BH{B}$ be a unital CP map.
\begin{enumerate}
\item \label{it:schwarz}
The Schwarz inequality
\begin{equation}
	\Lambda (a^\dag a) \geq \Lambda (a^\dag) \Lambda(a)
	\label{eq:schwarz}
\end{equation}
holds for all $a \in \BH{A}$. Here, for $a,b\in \BH{A}$, $a \geq b$ (or $b\leq a$) means that $a-b$ is a non-negative operator.
\item \label{it:md}
For arbitrary $a \in \BH{A}$, $a\in \mathcal{M}_\Lambda$ holds if and only if
\begin{equation}
	\Lambda(ba) = \Lambda(b) \Lambda(a), \quad \Lambda(ab) = \Lambda (a) \Lambda(b)
\end{equation}
holds for all $b\in \BH{A}$.
\end{enumerate}
\end{lemm}

\begin{proof}[Proof of Theorem~\ref{thm:main1g}]
Throughout the proof, we fix an arbitrary measurable set $X\in \Sigma$.
We define a unital CP map $\Lambda \colon \mathbf{B}(\cH_\Spr \otimes \cH_\Ppr) \to \BH{S}$ by 
\begin{equation}
	\Lambda (a) := \tr_P [  (\unit_S \otimes \rho_P) \Uspsp^\dag a \Uspsp  ]
	\quad (a\in \mathbf{B}(\cH_\Spr \otimes \cH_\Ppr)).
	\label{eq:Ldef}
\end{equation}
Since the measurement model $\mathbb{M}$ implements the PVM $(\Omega ,\Sigma, \sfE_S )$, we have
\begin{align}
	\sfE_S (X) &= \tr_P[(\unit_S \otimes \rho_P) \Uspsp^\dag (\unit_\Spr \otimes \sfF_\Ppr (X)) \Uspsp] \\
	&= \Lambda (\unit_\Spr \otimes \sfF_\Ppr (X)) .
	\label{eq:LaF}
\end{align}
Then we obtain
\begin{align}
	\Lambda (\unit_\Spr \otimes \sfF_\Ppr (X))^2
	&= \sfE_S(X)^2 
	\quad (\because \text{Eq.~\eqref{eq:LaF}}) 
	\\
	&= \sfE_S(X) \quad (\because \text{$\sfE_S(X)$ is a projection}) \\
	&=  \Lambda (\unit_\Spr \otimes \sfF_\Ppr (X)) 
	\quad (\because \text{Eq.~\eqref{eq:LaF}}) 
	\\
	&\geq \Lambda ( (\unit_\Spr \otimes \sfF_\Ppr (X))^2)
	\quad (\because \unit_\Spr \otimes \sfF_\Ppr (X) \geq (\unit_\Spr \otimes \sfF_\Ppr (X))^2) 
	\\
	&\geq \Lambda (\unit_\Spr \otimes \sfF_\Ppr (X))^2
	\quad (\because \text{Schwarz inequality \eqref{eq:schwarz}}) .
	\label{eq:LaF2}
\end{align}
This implies that equality of \eqref{eq:LaF2} holds and hence $\unit_\Spr \otimes \sfF_\Ppr (X) \in \mathcal{M}_\Lambda$.
From the $G$-invariance condition~\eqref{eq:Ginv}, for every $g\in G$ we have
\begin{align}
	\Lambda (U^\Spr_g \otimes U^\Ppr_g)
	&= \tr_P [  (\unit_S \otimes \rho_P) \Uspsp^\dag (U^\Spr_g \otimes U^\Ppr_g)   \Uspsp  ]
	\\
	&=  \tr_P [  (\unit_S \otimes \rho_P) U^S_g \otimes U^P_g ] \\
	&= \tr [\rho_P U^P_g] U^S_g .
	\label{eq:usg}
\end{align}
Thus for every $g\in G$ we obtain
\begin{align}
	\tr [\rho_P U^P_g]  \sfE_S(X)  U^S_g
	&= \Lambda  (\unit_\Spr \otimes \sfF_\Ppr (X)) \Lambda (U^\Spr_g \otimes U^\Ppr_g)
	\quad (\because \text{Eqs.~\eqref{eq:LaF} and \eqref{eq:usg}}) \\
	&= \Lambda  ((\unit_\Spr \otimes \sfF_\Ppr (X) ) (U^\Spr_g \otimes U^\Ppr_g ) )
	\quad (\because \unit_\Spr \otimes \sfF_\Ppr (X) \in \mathcal{M}_\Lambda ) \\
	&= \Lambda  ( (U^\Spr_g \otimes U^\Ppr_g ) (\unit_\Spr \otimes \sfF_\Ppr (X) ))
	\quad (\because \text{the generalized Yanase condition}) \\
	&=\Lambda  (U^\Spr_g \otimes U^\Ppr_g ) \Lambda (\unit_\Spr \otimes \sfF_\Ppr (X) )
	\quad (\because \unit_\Spr \otimes \sfF_\Ppr (X) \in \mathcal{M}_\Lambda )  \\
	&= \tr [\rho_P U^P_g]  U^S_g \sfE_S(X)  \quad  (\because \text{Eqs.~\eqref{eq:LaF} and \eqref{eq:usg}}). \label{eq:almostcom}
\end{align}
From the strong continuity of the unitary representation $(U_g^P)_{g\in G}$, the function $G\ni g \mapsto \tr [\rho_P U^P_g]  \in \cmplx $ is continuous.
Therefore, since $ \tr [\rho_P U^P_e] = 1 \neq 0$, there exists an open neighborhood $W \subseteq G$ of the unit element $e\in G$ such that $ \tr [\rho_P U^P_g] \neq 0$ for all $g\in W$.
Thus Eq.~\eqref{eq:almostcom} implies
\begin{equation}
	 \sfE_S(X)  U^S_g =  U^S_g  \sfE_S(X) \quad ( g\in W).
	 \label{eq:commW}
\end{equation}

We now define 
\begin{equation}
	H := \{ g \in G : \sfE_S(X)  U^S_g =  U^S_g  \sfE_S(X) \}
\end{equation}
and show that $H$ is an open subgroup of $G$.
We have $e \in H$ because $U^S_e = \unit_S$ commutes with $\sfE_S (X)$.
For every $g,h \in H$ we have
\begin{align}
	\sfE_S(X) U^S_{gh^{-1}}
	&= \sfE_S(X) U^S_g U^{S\dag}_{h} \\
	&= U^S_g  \sfE_S(X) U^{S\dag}_{h} \\
	&= U^S_g  (U^{S}_{h} \sfE_S(X) )^\dag \\
	&=  U^S_g  ( \sfE_S(X)  U^{S}_{h}  )^\dag \\
	&=  U^S_g U^{S\dag}_{h}  \sfE_S(X) \\
	&= U^S_{gh^{-1}}   \sfE_S(X) ,
\end{align}
which implies $gh^{-1} \in H$.
Thus, $H$ is a subgroup of $G$.
From Eq.~\eqref{eq:commW} we have $W \subseteq H$ and therefore $gW := \{gh : h \in W \} \subseteq H$ for every $g \in H$ because $H$ is a subgroup of $G$.
Since $gW$ is an open neighborhood of $g\in H$, this shows that $H$ is open.
Thus, $H$ is an open subgroup of $G$.
Since every connected topological group has no proper open subgroup (\cite{higgins_1975}, Section~II.7), we should have $G= H $, which proves the claim.
\end{proof}

%\section{Proof of Theorem~\ref{thm:main2g}} \label{sec:proof2}
%In this section we prove Theorem~\ref{thm:main2g}.
\begin{proof}[Proof of Theorem~\ref{thm:main2g}]
As in the proof Theorem~\ref{thm:main1g}, we define a unital CP map $\Lambda \colon \mathbf{B}(\cH_\Spr \otimes \cH_\Ppr) \to \BH{S}$ by Eq.~\eqref{eq:Ldef}. 
Since the system-environment model $\mathbb{M}$ implements $\mathcal{U}_{S\to \Spr}$, we have
%\begin{equation}
%	U_{S\to \Spr}\rho_S U_{S\to \Spr}^\dag = \tr_{\Ppr} [ \Uspsp (\rho_S \otimes \rho_P) \Uspsp^\dag ] \quad (\rho_S \in \TrH{S}),
%\end{equation}
%or equivalently in the Heisenberg picture
\begin{align}
	U_{S\to \Spr}^\dag a U_{S\to \Spr}
	&= \tr_{P} [  (\unit_S \otimes \rho_P) \Uspsp^\dag (a\otimes \unit_\Ppr) \Uspsp] 
	\\
	&= \Lambda (a \otimes \unit_\Ppr)
	\quad ( a \in \BH{\Spr}). \label{eq:aLa}
\end{align}
Therefore for every $a\in \BH{\Spr}$ we have
\begin{align}
	\Lambda (a^\dag \otimes \unit_\Ppr) \Lambda (a \otimes \unit_\Ppr)
	&= U_{S\to \Spr}^\dag a^\dag U_{S\to \Spr} U_{S\to \Spr}^\dag a U_{S\to \Spr} \\
	&= U_{S\to \Spr}^\dag a^\dag a U_{S\to \Spr} \\
	&= \Lambda (a^\dag a \otimes \unit_\Ppr ) \\
	&= \Lambda ((a^\dag \otimes \unit_\Ppr)(a \otimes \unit_\Ppr))
\end{align}
and, by interchanging $a$ and $a^\dag$, we also have
\begin{equation}
	\Lambda (a \otimes \unit_\Ppr) \Lambda (a^\dag \otimes \unit_\Ppr)
	=  \Lambda ((a \otimes \unit_\Ppr)(a^\dag \otimes \unit_\Ppr)).
\end{equation}
Thus, we obtain $a \otimes \unit_{\Ppr} \in \mathcal{M}_\Lambda$ for every $a \in \BH{\Spr}$.

Now for every $a\in \BH{S}$ and every $b \in \BH{\Ppr}$ we have
\begin{align}
	\Lambda (\unit_\Spr \otimes b) a
	&= \Lambda (\unit_\Spr \otimes b) U_{S\to \Spr}^\dag a^\prime U_{S\to \Spr}
	\quad (a^\prime := U_{S\to \Spr}a U_{S\to \Spr}^\dag \in \BH{\Spr}) \\
	&= \Lambda (\unit_\Spr \otimes b)  \Lambda (a^\prime \otimes \unit_\Ppr)
	\quad (\because \text{Eq.~\eqref{eq:aLa}}) \\
	&= \Lambda ((\unit_\Spr \otimes b) (a^\prime \otimes \unit_\Ppr) )
	\quad (\because a^\prime \otimes \unit_\Ppr \in \mathcal{M}_\Lambda) \\
	&=\Lambda ( (a^\prime \otimes \unit_\Ppr) (\unit_\Spr \otimes b)) \\
	&= \Lambda  (a^\prime \otimes \unit_\Ppr)  \Lambda (\unit_\Spr \otimes b)
	\quad (\because a^\prime \otimes \unit_\Ppr \in \mathcal{M}_\Lambda) \\
	&= a \Lambda (\unit_\Spr \otimes b).
\end{align}
This implies $\Lambda (\unit_\Spr \otimes b) \in \BH{S}^\prime = \cmplx \unit_S$, 
where the prime denotes the commutant and $\cmplx \unit_S := \{ c\unit_S : c \in \cmplx\} $.
Therefore, there exists a linear functional $\psi \colon \BH{\Ppr} \to \cmplx$ such that
\begin{equation}
	\Lambda (\unit_\Spr \otimes b) = \psi(b) \unit_S 
	\quad (b\in \BH{\Ppr}) . 
	\label{eq:psib}
\end{equation}
Thus, for every $g\in G$ we have
\begin{align}
	\Lambda (U_g^\Spr \otimes U_g^\Ppr) 
	&= \Lambda ((U_g^\Spr \otimes \unit_\Ppr) (\unit_\Spr \otimes U_g^\Ppr))
	\\
	&= \Lambda (U_g^\Spr \otimes \unit_\Ppr) \Lambda (\unit_\Spr \otimes U_g^\Ppr) 
	\quad (\because U_g^\Spr \otimes \unit_\Ppr \in \mathcal{M}_\Lambda)
	\\
	&= \psi(U_g^\Ppr)  U_{S\to \Spr}^\dag U_g^\Spr U_{S\to \Spr} 
	\quad (\because \text{Eqs.~\eqref{eq:aLa} and \eqref{eq:psib}}).
	\label{eq:LaUg1}
\end{align}
On the other hand, $\Lambda (U_g^\Spr \otimes U_g^\Ppr)$ can also be written as
\begin{align}
	\Lambda (U_g^\Spr \otimes U_g^\Ppr) 
	&= \tr_{P} [  (\unit_S \otimes \rho_P) \Uspsp^\dag (U_g^\Spr \otimes U_g^\Ppr) \Uspsp] 
	\\
	&=  \tr_{P} [  (\unit_S \otimes \rho_P) (U_g^S \otimes U_g^P)] \quad (\because \text{$G$-invariance condition \eqref{eq:Ginv}}) 
	\\
	&= \phi(U_g^P) U_g^S , \label{eq:LaUg2}
\end{align}
where $\phi(b) := \tr [\rho_P b]$ $(b\in \BH{P})$.
Hence from Eqs.~\eqref{eq:LaUg1} and \eqref{eq:LaUg2} we obtain
\begin{equation}
	\phi(U_g^P) U^S_g  = \psi(U_g^\Ppr) V^S_g
	\quad (g\in G),
	\label{eq:phipsi}
\end{equation}
where we defined a strongly continuous unitary representation $\gin{V^S_g}$ of $G$ acting on $\cH_S$ by
\begin{equation}
	V_g^S :=  U_{S\to \Spr}^\dag U_g^\Spr U_{S\to \Spr}  \quad (g\in G).
\end{equation}
By the strong continuity of $(U_g^P)_{g\in G}$, the function $G\ni g \mapsto \phi(U_g^P) \in \cmplx$ is continuous.
Thus, from $\phi(U_e^S)  =1 \neq 0$, there exists an open neighborhood $W \subseteq G$ of $e\in G$ such that $\phi(U_g^S) \neq 0$ for all $g\in W$.
Hence, by noting $U^S_g \neq 0$, Eq.~\eqref{eq:phipsi} implies that for every $g\in W$ there exists a scalar $c \in \cmplx$ such that
\begin{equation}
	V_g^S = c U_g^S . 
	\label{eq:VU}
\end{equation}
We now define 
\begin{equation}
	H := \{ g\in G : \exists  c\in \cmplx  \text{ s.t.\ }  V_g^S = c U_g^S   \}
\end{equation}
and prove that $G=H$.
Since $V_e^S =  U_e^S = \unit_S $, we have obviously $e \in H$.
Let $g,h \in H$ be arbitrary elements and take $c_1,c_2 \in \cmplx$ such that 
$V_g^S = c_1U_g^S$ and $V_h^S = c_2 U_h^S$.
Then we have
\begin{align}
	V_{gh^{-1}}^S &= V_g^S V_h^{S\dag} \\
	&= c_1 \overline{c_2} U^S_g U^{S\dag}_h
	\\&= c_1 \overline{c_2} U^S_{gh^{-1}},
\end{align}
which implies $gh^{-1} \in H$.
Therefore, $H$ is a subgroup of $G$.
Moreover, from the construction of $W$ we have $W \subseteq H$.
Since $W$ is an open neighborhood of $e$, the connectedness of $G$ implies $G=H$ as in the proof of Theorem~\ref{thm:main1g}.
Thus, for every $g\in G$ there exists a scalar $c_g\in \cmplx$ such that 
\begin{equation}
	V_g^S = c_g U_g^S.  \label{eq:VcgU} 
\end{equation}

We establish the theorem by proving that $(c_g)_{g\in G}$ a $1$-dimensional continuous unitary representation of $G$.
Since for every $g\in G$ we have
\begin{equation}
	c_g \unit_S = V_g^S U_{g^{-1}}^S ,
	\label{eq:cgVU}
\end{equation}
the scalar $c_g$ is unique and $c_g \in \Uone$.
From the strong continuity of $\gin{U^S_g}$ and $\gin{V^S_g}$,
Eq.~\eqref{eq:cgVU} also implies that $G\ni g \mapsto c_g \in \Uone$ is continuous.
Thus, we only have to show that $G \ni g \mapsto c_g \in \Uone$ is a group homomorphism.
From $V_e^S =  U_e^S = \unit_S $ we have $c_e = 1$.
For every $g,h \in G$ we have
\begin{align}
	c_{gh^{-1}} U^S_{gh^{-1}} &= V^{S}_{gh^{-1}} \\
	&= V^S_g V^{S\dag }_h \\
	&= c_g \overline{c_h} U^S_g U^{S\dag }_h
	\\
	&= c_g \overline{c_h} U^S_{gh^{-1}} ,
\end{align}
which implies $c_{gh^{-1}} = c_g \overline{c_h}$.
Hence $G \ni g \mapsto c_g \in \Uone$ is a group homomorphism, which completes the proof.
\end{proof}

\section{A brief review of the WAY-type trade-off relations between the implementation error and the required resource in quantum processes}\label{sec:review}
In the main text, we have shown the no-go theorems for error-free implementations of projective measurements and unitary gates.
As briefly introduced in the Introduction, it is known that trade-off relations between implementation error and the resources required for implementation have been established when approximate implementations are allowed.
This field has developed actively in recent years, and results have also been obtained for general quantum gates, not limited to projective measurements and unitary gates.
For this reason, we give a brief review below. In the following review, the results are for finite-dimensional systems unless otherwise noted.

The studies of how conservation laws restrict quantum information processing with errors was started by M.\ M.\ Yanase \cite{PhysRev.123.666}.
In 1961, Yanase derived a trade-off inequality that shows that the size (dimension) of the implementation device is inversely proportional to the implementation error, in the same setup as the WAY theorem (Figure \ref{fig:2_supp}).
In 2002, M.\ Ozawa further developed this result in two directions.
In the first direction, Ozawa showed a trade-off inequality for the WAY-theorem setup with the Yanase condition \cite{OzawaWAY}. The inequality, called the WAY-Ozawa inequality, shows that the implementation error is inversely proportional to the variance of the conserved quantity $L_P$ in the implementation device $P$.
In the second direction, Ozawa showed that a very similar inequality to the WAY-Ozawa inequality holds for the implementation of the Controlled-NOT (C-NOT) gate \cite{ozawaWAY_CNOT}.
The inequality shows that in an arbitrary implementation of the C-NOT gate on 2qubit with the $Z$-spin conservation law (the computational basis is the eigenstates of the $Z$-spins), the implementation error of the C-NOT gate is inversely proportional to the variance of the conserved quantity ($Z$-direction spin in $P$) in the implementation device (Figure \ref{fig:1_supp}).

Ozawa's inequality for the C-NOT gate is very suggestive, since if the inequality could be extended to arbitrary unitary gates in general $d$-level systems, very similar trade-off relations could be established for the very different quantum information processing, i.e.\ projective measurements and unitary gates.
This problem has been open for a long time, but after similar results were shown for various limited unitary gates \cite{ozawa2003uncertainty,PhysRevA.75.032324,Karasawa_2009}, it was shown for arbitrary unitary gates in \cite{TSS}.
Ref.~\cite{TSS} also shows that the variance of the conserved quantity can be substituted by the quantum fluctuation of the conserved quantity (quantum Fisher information). Since the quantum Fisher information is the standard measure in the resource theory of asymmetry \cite{skew_resource,Takagi_skew,YT,Marvian_distillation,Hansen,kudo_fisher_2022} and related to the quantum fluctuation \cite{min_V_Yu,min_V_Petz,kudo_fisher_2022}, it was shown that the trade-off between the implementation error and the required resource is generally valid for unitary gates.
This is paired with extensions of the WAY-Ozawa inequality,  trade-off relations between the measurement error and the required resource (quantified by the quantum Fisher information) \cite{korzekwa2013,TN}.

After the results in Ref. \cite{TSS}, there were active developments regarding the trade-off between the implemenation error and the required resources in the implementation of unitary gates.
First, the minimum sufficient amount of required resources to achieve a given implementation error was given in the form of an asymptotic equality \cite{TSS2}.
Later, by similar methods to Ref. \cite{TSS}, the tradeoff relations between error and resource requirements in a broader class of resource theories were given \cite{TT,Yuxiang1}.
Recently, these results were further extended to infinite-dimensional systems \cite{Yuxiang2}.
All of these results provide in common inversely proportional relationships between resource requirements and errors in unitary gate implementations by using similar methods.

Furthermore, attempts have been made to unify the WAY-type restrictions given to various quantum information processing beyond the unitary gates and projective measurements.
In this direction, the unification of the WAY-type theorems for unitary gates, error correcting codes, and information recovery from black holes was done \cite{TS}. 
This result unifies the WAY-type theorems for unitary gates given in Refs.\cite{TSS,TSS2}, and the approximate Eastin-Knill theorem as the inversely proportional relation between decoding error and the number of code qubits for transversal codes.
After that, a similar unification result is also given for unitary gates and error-correcting codes \cite{liu_quantum_2022}.
And more recently, a theorem was given to unify more variants of the WAY-type restrictions \cite{arxiv.2206.11086}.
This result unifies the WAY-type theorems for projective measurement and unitary gates, and the approximate Eastin-Knill theorems as corollaries. In particular, for the measurement, a WAY-Ozawa type inequality was derived for the error defined by the fidelity error of the measurement output, rather than the error defined by the noise operator \cite{OzawaWAY} as before. Also, the result in \cite{arxiv.2206.11086} gives the trade-off relation between entropy production and the required coherence of the implementation of arbitrary Gibbs-preserving maps, and gives a general restriction on the classical/quantum information recovery from the black hole with the energy conservation law.

\begin{figure}
\centering
% {\large
% \Qcircuit @C=1em @R=1.2em { \lstick{\rho_S} &    \ustick{S}\qw& \multigate{1}{U^\dag} & \ustick{\Spr}  \qw &   \rstick{\text{trash}} \qw \\ 
% \lstick{\rho_P} &    \dstick{P} \qw & \ghost{U^\dag}& \dstick{\Ppr} \qw  & \measureD{\sfF_{\Ppr}} \qw}
% }
\includegraphics[width=6cm,clip]{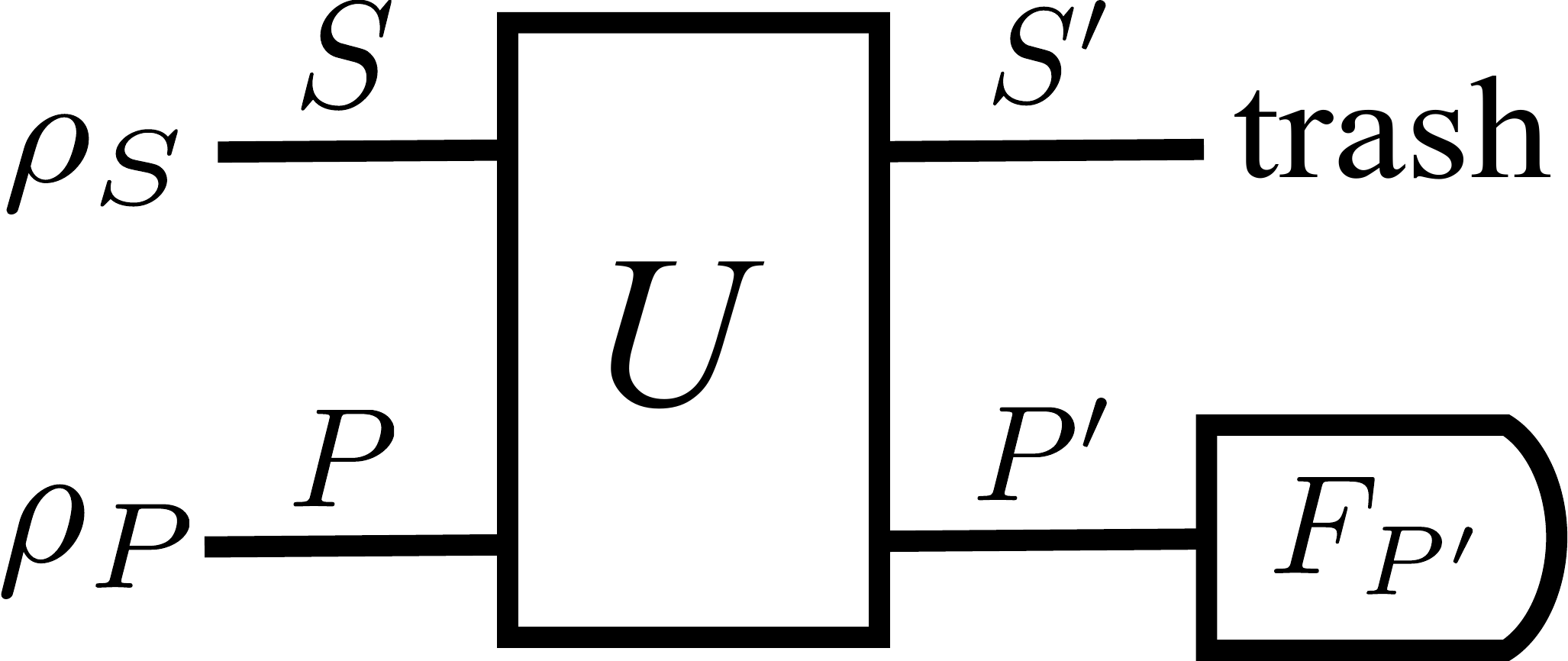}
% $=$
% \includegraphics[width=3cm,clip]{fig2_2.pdf}
\caption{A measurement model that implements a POVM $\sfE_S$. When we assume the Yanase condition, the measurement operators of $F_{P'}$ have to commute with the conserved quantity $L_{P'}$ on $P'$. 
When we try to implement a projective measurement $\calP$ approximately, the implementation errors are defined in two ways.
The first way is to define the error as the expectation value of the square of the \lq\lq{}noise operator,\rq\rq{} which is used in Refs.\cite{OzawaWAY,korzekwa2013,TN}.
The second way is to define the the gate-fidelity distance between the desired measurement $\calP$ and the actually implemented measurement channel  $\calM$, which is used in Ref.\cite{arxiv.2206.11086}.}
\label{fig:2_supp}
\end{figure}

\begin{figure} 
\centering
\includegraphics[width=6cm,clip]{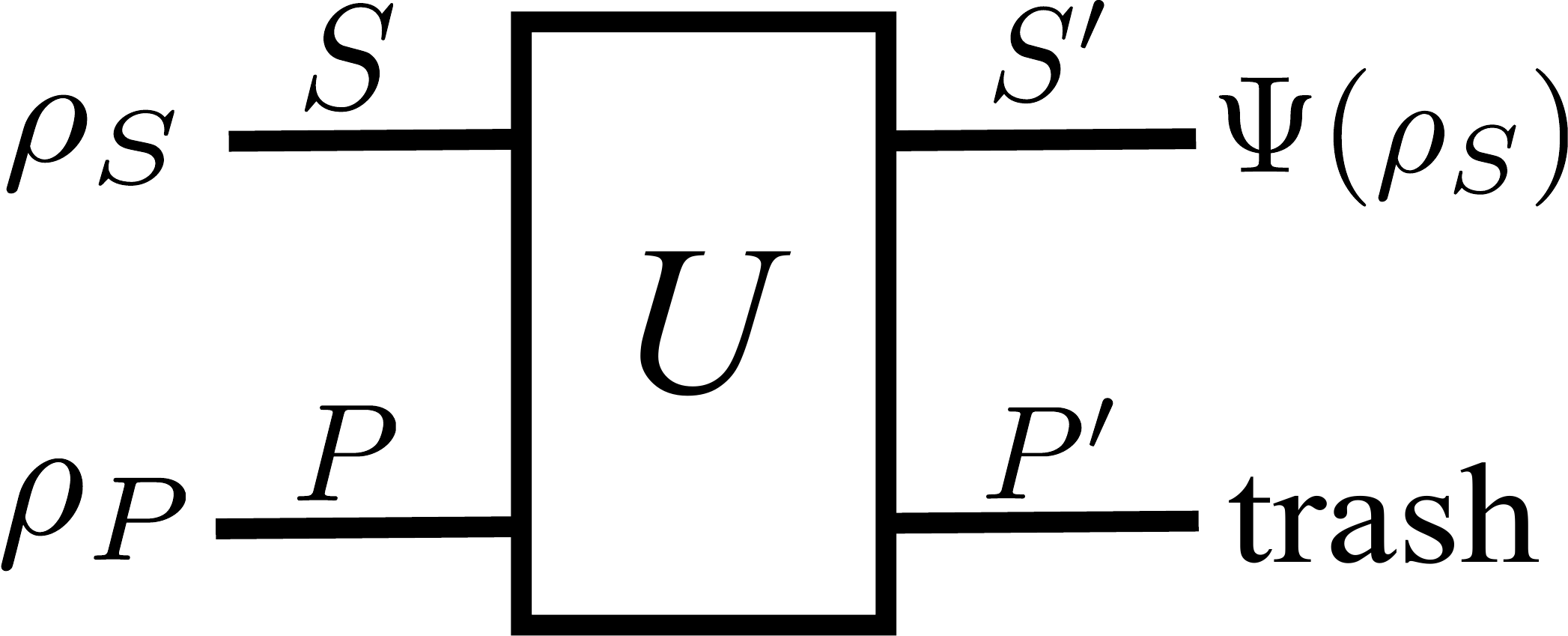}
% {\large
% \Qcircuit @C=1em @R=1.2em { \lstick{\rho_S} &    \ustick{S}\qw& \multigate{1}{U^\dag} & \ustick{\Spr} \qw  &   \rstick{\Psi(\rho_S)} \qw \\ 
% \lstick{\rho_P} &    \dstick{P} \qw & \ghost{U^\dag}& \dstick{\Ppr} \qw  & \rstick{\text{trash}} \qw}
%  }
\caption{A system-environment model that implements a channel $\Psi$. When we try to implement a unitary channel $\calU$ approximately, the implementation error is defined the gate-fidelity distance (or the entanglement gate-fidelity distance) between $\Psi$ and $\calU$.}
\label{fig:1_supp}
\end{figure}

\end{widetext}
\end{document}